\documentclass{lmcs}
\pdfoutput=1
\usepackage[utf8]{inputenc}

\usepackage{lastpage}
\lmcsdoi{21}{2}{6}
\lmcsheading{}{\pageref{LastPage}}{}{}%
{Apr.~08,~2024}{Apr.~21,~2025}{}


\usepackage{amsmath,latexsym,amssymb,amsfonts,amscd,amsthm,stmaryrd,textcomp}
\usepackage{commath}
\usepackage[noend,linesnumbered,ruled,vlined,resetcount]{algorithm2e}
\usepackage{hyperref}
\usepackage{cleveref}
\Crefname{exa}{Example}{Examples}
\Crefname{thm}{Theorem}{Theorems}
\Crefname{lem}{Lemma}{Lemmas}
\Crefname{defi}{Definition}{Definitions}
\usepackage[]{todonotes}
\usepackage{bm}
\usepackage{blkarray, bigstrut}
\usepackage{mfirstuc}
\usepackage{tikz, pgfplots}
\usetikzlibrary{%
  intersections,%
  calc,fit,tikzmark,%
  arrows.meta,%
  decorations.pathreplacing,%
	decorations.pathmorphing,%
	decorations.markings
}
\usepackage{rwthcolors}
\usepackage{subcaption}
\usepackage{booktabs}

\newcommand{\vect}[1]{\bm{#1}}
\newcommand{\minus}{\text{-}}
\newcommand{\rvect}[2]{\vect{#1}_{#2,{\minus}}}
\newcommand{\cvect}[2]{\vect{#1}_{{\minus},#2}}
\newcommand{\matr}[1]{\capitalisewords{#1}}
\newcommand{\sys}[2]{\matr{#1}\vect{x} \leq \vect{#2}}
\newcommand{\submatr}[2]{\matr{#1}\text{\smaller{$[#2]$}}}
\newcommand{\subvect}[2]{\vect{#1}\text{\smaller{$[#2]$}}}
\newcommand{\subsys}[3]{\submatr{#1}{#3}\vect{x} \leq \subvect{#2}{#3}}
\newcommand{\subsyseq}[3]{\submatr{#1}{#3}\vect{x} = \subvect{#2}{#3}}

\newcommand{\bnd}[2]{\textit{bnd}_{#1}(#2)}

\newcommand{\lbs}[2]{I_{#2}^{-}(\matr{#1})}
\newcommand{\ubs}[2]{I_{#2}^{+}(\matr{#1})}
\newcommand{\nbs}[2]{I_{#2}^{0}(\matr{#1})}
\newcommand{\sbs}[2]{I_{#2}^{*}(\matr{#1})}

\newcommand{\nats}{\mathbb{N}}
\newcommand{\reals}{\mathbb{R}}
\newcommand{\rationals}{\mathbb{Q}}

\newcommand{\solset}[1]{\textit{sol}(#1)}

\newcommand{\rank}[1]{\textit{rank}(#1)}

\newcommand{\fmpproj}[3]{P_{#2,#3}(#1)}

\newcommand{\proj}[2]{{#1}|_{#2}}

\newcommand{\sat}[0]{\textit{SAT}}
\newcommand{\unsat}[0]{\textit{UNSAT}}
\newcommand{\partialunsat}[0]{\textit{PARTIAL-UNSAT}}

\newcommand\fmplexelim[0]{\textnormal{\texttt{FMP}}}

\newcommand{\dom}{\textit{dom}}

\newcommand{\vs}{/\!\!/}

\newcommand{\chngd}[1]{\setlength{\fboxrule}{1.5pt}\fcolorbox{black!25}{black!25}{#1}}
\newcommand{\chng}[1]{\setlength{\fboxrule}{1.5pt}\fcolorbox{black!25}{white}{#1}}

\newcommand{\boxchngd}[2]{\tikz[remember picture,overlay]{\node[yshift=7pt,fill=black!25,fit={(pic cs:A#1)($(pic cs:B#1)+(#2\linewidth,.5\baselineskip)$)}] {};}}
\newcommand{\boxchngdstart}[1]{\tikzmark{A#1}}
\newcommand{\boxchngdend}[1]{\tikzmark{B#1}}

\newcommand{\refalgobase}{Algorithm~\hyperref[algo:combined-base]{2a} }
\newcommand{\refalgobounds}{Algorithm~\hyperref[algo:leaving-out-bounds]{2b} }
\newcommand{\refalgobacktracking}{Algorithm~\hyperref[algo:backtracking]{2c} }

\begin{document}

\title[FMplex: Exploring a Bridge between Fourier-Motzkin and Simplex]{FMplex: Exploring a Bridge\texorpdfstring{\\}{} between Fourier-Motzkin and Simplex}

\thanks{%
	Jasper Nalbach and Valentin Promies were supported by the Deutsche Forschungsgemeinschaft (DFG) as part of AB 461/9-1 \emph{SMT-ART}.
	Jasper Nalbach was also supported by the DFG RTG 2236 UnRAVeL
}

\author[V.~Promies]{Valentin Promies\lmcsorcid{0000-0002-3086-9976}}[a]
\author[J.~Nalbach]{Jasper Nalbach\lmcsorcid{0000-0002-2641-1380}}[a]
\author[E.~Ábrahám]{Erika Ábrahám\lmcsorcid{0000-0002-5647-6134}}[a]
\author[P.~Kobialka]{Paul Kobialka\lmcsorcid{0000-0002-0635-1915}}[b]

\address{RWTH Aachen University, Germany}
\email{nalbach@cs.rwth-aachen.de, promies@cs.rwth-aachen.de, abraham@cs.rwth-aachen.de}
\address{University of Oslo, Norway}
\email{paulkobi@ifi.uio.no}

\begin{abstract}
  In this paper we present a quantifier elimination method for conjunctions of  \emph{linear real arithmetic} constraints.
  Our algorithm is based on the \emph{Fourier-Motzkin variable elimination} procedure, but by case splitting we are able to reduce the worst-case complexity from doubly to singly exponential.
  The adaption of the procedure for SMT solving has strong correspondence to the \emph{simplex algorithm}, therefore we name it \emph{FMplex}.
  Besides the theoretical foundations, we provide an experimental evaluation in the context of SMT solving.
  This is an extended version of the authors' work previously published at the fourteenth International Symposium on Games, Automata, Logics, and Formal Verification (GandALF 2023).
\end{abstract}

\maketitle
\section{Introduction}
\emph{Linear real arithmetic (LRA)} is a powerful first-order theory with strong practical relevance.
We focus on checking the satisfiability of \emph{conjunctions} of LRA constraints, which is needed e.g. for solving quantifier-free LRA formulas using \emph{satisfiability modulo theories (SMT) solvers}.
This problem is known to be solvable in \emph{polynomial} worst-case complexity but, surprisingly, the \emph{ellipsoid} method \cite{KHACHIYAN198053} proposed in 1980 by Khachiyan is still the only available algorithm that implements this bound.
However, this method is seldom used in practice due to its high average-case effort.
Instead, most approaches employ the \emph{simplex} algorithm introduced by Dantzig in 1947, which is quite efficient in practice, despite a \emph{singly exponential} worst case complexity.
A third available solution is the \emph{Fourier-Motzkin variable elimination (FM)} method, proposed in 1827 by Fourier \cite{fourier1827analyse} and re-discovered in 1936 by Motzkin \cite{motzkin1936beitrage}.
In contrast to the other two approaches, FM admits quantifier elimination, but it has a \emph{doubly exponential} worst case complexity, even though there have been various efforts to improve its efficiency by detecting and avoiding redundant computations (e.g. \cite{imbert1993fourier,JingComplexityFME}).

In recent work \cite{Nalbach_2023}, we introduced a novel method, which is derived from the FM method, but which turns out to have striking resemblance to the simplex algorithm.
This yields interesting theoretical insights into the relation of the two established methods and the nature of the problem itself.
Here, we provide an extended version of \cite{Nalbach_2023}, and our contributions naturally include those of the original paper:
\begin{itemize}
    \item The presentation of \emph{FMplex}, a new variable elimination method based on a divide-and-conquer approach. We show that it does not contain certain redundancies Fourier-Motzkin might generate and lowers the overall complexity from \emph{doubly} to \emph{singly} exponential.
    \item An adaptation of FMplex for SMT solving, including methods to prune the search tree based on structural observations.
    \item A theorem formalizing connections between FMplex and the simplex algorithm.
    \item An implementation of the SMT adaptation and its experimental evaluation.
\end{itemize}

\noindent We extend our previous work by the following novel contributions:
\begin{itemize}
	\item Additional and more detailed explanations, thereby improving readability and making some insights more explicit.
	In particular, we explain in \Cref{sec:varelim} how the Fourier-Motzkin method still has advantages when only few variables are eliminated, and we added \Cref{sec:strict} discussing how to handle strict constraints.
	\item Extensive and reworked examples illustrating more aspects and intricacies of the algorithms.
        In particular, \Cref{def:redundancy} and \Cref{example:redundancies} concerning redundancies, as well as \Cref{example:worstcase} demonstrating the worst-case complexity are new. Moreover, \Cref{example:fmplex-elim-procedure,example:nonbasis,example:backtracking}  (Examples 2, 3, 4 in \cite{Nalbach_2023}) have been extended and enriched with illustrations.
	\item Full and rigorous proofs for the central \Cref{thm:global-conflict,thm:basis-nonbasis-correspondence,thm:unique-base-termination} (Theorems 6, 7, 10 in \cite{Nalbach_2023}). These proofs are not trivial, and they provide interesting insights into the problem's nature.
	\item We improved readability of some notation in \Cref{subsec:avoiding-redundancies} as well as some notation for matrices.
\end{itemize}
After recalling necessary preliminaries in \Cref{sec:preliminaries}, we introduce our FMplex method first for variable elimination in \Cref{sec:varelim} and then for SMT solving in \Cref{sec:sat}.
We present related work and compare FMplex with other methods, first qualitatively in \Cref{sec:othermethods}, and then experimentally in \Cref{sec:experiments}.
We discuss future work and conclude the paper in \Cref{sec:conclusion}.

\section{Preliminaries}\label{sec:preliminaries}
Let $\reals$, $\rationals$ and $\nats$ denote the set of real, rational respectively natural ($0\notin\nats$) numbers.
For $k \in \nats$ we define $[k]:=\{1,\ldots,k\}$.
Throughout this paper, we fix $n \in \nats$, a set $X = \{ x_1,\ldots,x_n \}$ and a corresponding vector $\vect{x} = (x_1,\ldots,x_n)^T$ of $\reals$-valued variables.

\paragraph{Matrices}
We use bold lower-case letters (e.g.~$\vect{f}$) to denote vectors and upper case letters (e.g.~$\matr{A}$) for matrices.
For any $m \in \nats$ let $\matr{e} \in \rationals^{m \times m}$ be the identity matrix, and let $\vect{0} = (0\ \cdots\ 0)^T \in\rationals^{m\times 1}$ be the zero vector. The dimensions of $E$ and $\vect{0}$ will be clear from the context.
The $i$-th component of $\vect{f}\in \rationals^{m\times 1}\cup\rationals^{1\times m}$ is denoted by $f_i$ and the component-wise comparison to zero by $\vect{f} \geq 0$.

For $\matr{a} \in \rationals^{m \times n}$, $\rvect{a}{i} \in \rationals^{1 \times n}$ and $\cvect{a}{i} \in \rationals^{m \times 1}$ denote the $i$-th row respectively column vector of $\matr{a}$.
Furthermore, $\submatr{a}{I}$ denotes the sub-matrix of $\matr{a}$ containing only the rows with indices from some $I\subseteq [m]$.
That is, the rows of $\submatr{a}{I} \in \rationals^{|I|\times n}$ are exactly the $\rvect{a}{i}$ with $i \in I$.

For $\vect{f} \in \rationals^{1\times m}$, $\vect{f}\matr{a}$ is a \emph{linear combination} of the rows $i \in [m]$ of $\matr{a}$ with $f_i \neq 0$.
We call $\matr{a}$ \emph{linearly independent} if none of its rows is a linear combination of its other rows, and \emph{linearly dependent} otherwise.
The \emph{rank of \matr{a}} $\rank{\matr{a}}$ is the maximal size of an $I \subseteq [m]$ so that $\submatr{a}{I}$ is linearly independent.

\paragraph{Linear Constraints}
Let $\vect{a} = (a_1, \ldots, a_n) \in \rationals^{1\times n}$, $b \in \rationals$ and $\sim \in \{ =,\leq,<,\neq \}$ a \emph{relation symbol}.
We call $\vect{a}\vect{x} := (a_1x_1 + \ldots + a_nx_n)$ a \emph{linear term} and $\vect{a}\vect{x} \sim b$ a \emph{linear constraint}, which is \emph{weak} if $\sim \in \{=,\leq\}$  and \emph{strict} otherwise.
A \emph{system of linear constraints}, or short a \emph{system}, is a non-empty finite set of linear constraints.
In most parts of this paper, we only consider constraints of the form $\vect{a}\vect{x} \leq b$.
We can write every system $C=\{\rvect{a}{i}\;\vect{x} \leq b_i \mid i \in [m]\}$ of such constraints in \emph{matrix representation} $\sys{a}{b}$ with suitable $\matr{A} \in \rationals^{m \times n}$ and $\vect{b} \in \rationals^{m\times 1}$.
Conversely, every row $\rvect{a}{i}\; \vect{x} \leq b_i, \ i \in [m]$ of $\sys{A}{b}$ is a linear constraint.
Thus, the representations are mostly interchangeable; however, the matrix representation allows identical rows in contrast to the set notation.
As the latter will play a role later on, we will stick to the matrix representation.

\paragraph{Variable Assignments}
An \emph{assignment} is a function $\alpha: Y \to \reals$ with domain $\dom(\alpha)=Y\subseteq X$.
For an assignment $\alpha$, a variable $x_i$ and some $r \in \reals$, the \emph{extension} $\alpha[x_i \mapsto r]$ is the assignment with domain $\dom(\alpha)\cup\{x_i\}$ such that $\alpha[x_i \mapsto r](x_j)=\alpha(x_j)$ for all $x_j\in \dom(\alpha)\setminus\{x_i\}$ and $\alpha[x_i \mapsto r](x_i)=r$.
For $Z\subseteq Y$, the \emph{restriction} $\proj{\alpha}{Z}$ is the assignment with domain $Z$ such that $\proj{\alpha}{Z}(x_i)=\alpha(x_i)$ for all $x_i\in Z$.
We extend these notations to sets of assignments accordingly.

The standard \emph{evaluation} of a linear term $t$ under $\alpha$ is written $\alpha(t)$.
We say that $\alpha$ \emph{satisfies} (or is a solution of) a constraint $c = (\vect{a}\vect{x} \sim b)$ if $\alpha(\vect{a}\vect{x}) \sim b$ holds, and denote this fact by $\alpha \models c$.
All solutions of $c$ build its \emph{solution set} $\solset{c}$.
Similarly, $\alpha\models (\matr{a} \vect{x} \leq \vect{b})$ denotes that $\alpha$ is a common solution of all linear constraints in the system $\matr{a} \vect{x} \leq \vect{b}$.
A system is \emph{satisfiable} if it has a common solution, and \emph{unsatisfiable} otherwise.
Note that each satisfiable system has also a rational-valued solution.

\vspace{0.75em}
\noindent We will make use of the following two well-known results.

\begin{thm}[Farkas' Lemma \cite{farkas1902theorie}] \label{thm:farkas}
    Let $\matr{a} \in \rationals^{m \times n}$ and $\vect{b} \in \rationals^{m\times 1}$.
	Then the system $\matr{a} \vect{x} \leq \vect{b}$ is satisfiable if and only if for all $\vect{f} \in \rationals^{1\times m}$ with $\vect{f} \geq 0$ and $\vect{f} \matr{a} = \vect{0}$ it holds $\vect{f} \vect{b} \geq 0$.
\end{thm}

\begin{thm}[Fundamental Theorem of Linear Programming, as in \cite{luenberger1984linear}]\label{thm:linearprogramming}
	Assume $\matr{a} \in \rationals^{m \times n}$ and $\vect{b} \in \rationals^{m \times 1}$.
	Then $\sys{a}{b}$ is satisfiable if and only if there exists a subset $I \subseteq [m]$ such that $\submatr{a}{I}$ is linearly independent, $|I| = \rank{\matr{a}}$, and there exists an assignment $\alpha: X \to \reals$ with $\alpha \models (\sys{a}{b})$ and $\alpha \models (\submatr{a}{I}\vect{x} = \subvect{b}{I})$.
\end{thm} 

\subsection{Fourier-Motzkin Variable Elimination}
To eliminate any variable $x_j \in X$ from a system $\sys{a}{b}$, the \emph{Fourier-Motzkin} (FM) \cite{fourier1827analyse,motzkin1936beitrage} method computes another system\footnote[1]{
	Note that $\matr{a}$ and $\matr{a'}$ have the same number of columns, but the column $\cvect{a}{j}'$ corresponding to $x_j$ is zero. This presentation makes our later steps easier to formulate in terms of matrix multiplications.
} $\sys{a'}{b'}$ with $\cvect{a'}{j}=0$ and such that an assignment $\alpha$ is a solution of $\sys{a'}{b'}$ if and only if there is $r\in\rationals$ so that $\alpha[x_j\mapsto r]$ is a solution of $\sys{a}{b}$.
Graphically, the solution set of $\sys{a'}{b'}$ is the projection of the solutions of $\sys{a}{b}$ onto $X \setminus \{ x_j \}$.

The idea of the FM method is as follows.
For each $i \in [m]$ with $a_{i,j} \neq 0$, the constraint $\rvect{a}{i}\;\vect{x} \leq b_i$ can be rewritten as either a \emph{lower bound} or an \emph{upper bound} on $x_j$:
\begin{align*}
    x_j \sim \frac{1}{a_{i,j}}\big(b_i - \sum_{k \in [n] \setminus \{j\}} a_{i,k} x_k\big), \text{ where } \sim = \left.\begin{cases}
		\geq & \text{ if } a_{i,j} < 0\\
		\leq & \text{ if } a_{i,j} > 0
	\end{cases}\right.
\end{align*}

\noindent The term defining the bound is the same in both cases, and we denote it as $\bnd{j}{\rvect{a}{i}\;\vect{x} \leq b_i}$.\vspace{-1em}
\begin{defi}
	For $\matr{a} \in \rationals^{m \times n}$, we define the following three index sets:
    \begin{align*}
        \lbs{a}{j} := \{i \in [m] \mid a_{i,j} < 0\}\, \quad \ubs{a}{j} := \{i \in [m] \mid a_{i,j} > 0\}\, \quad \nbs{a}{j} := \{i \in [m] \mid a_{i,j} = 0\}
    \end{align*}
\end{defi}
\noindent $\lbs{a}{j}$, $\ubs{a}{j}$ and $\nbs{a}{j}$ indicate the rows of $\sys{a}{b}$ which induce lower bounds, upper bounds and no bounds on $x_j$, respectively.

Due to the density of the reals, there exists a value for $x_j$ that satisfies all bounds if and only if each lower bound is less than or equal to each upper bound.
However, in general the involved bounds are symbolic, and thus their values depend on the values of other variables.
Therefore, we cannot directly check whether the lower bounds are equal or below the upper bounds, but instead we express this by a new constraint set $\sys{a'}{b'}$, which is defined by the following:
\[
    \{\bnd{j}{\rvect{a}{\ell}\;\vect{x} \leq b_{\ell}} \leq \bnd{j}{\rvect{a}{u}\;\vect{x} \leq b_{u}} \mid (\ell,u) \in \lbs{a}{j}\times\ubs{a}{j}\}\  \cup\  \{\rvect{a}{i}\;\vect{x} \leq b_i \mid i \in \nbs{a}{j}\}
\]

\noindent Formulating this idea with respect to the matrix representation, the FM method applies the following transformation, using matrix multiplications. Recall that $\matr{e}$ is the identity matrix, and thus $\rvect{e}{i}$ is the $i$-th unit vector.

\begin{defi}[Fourier-Motzkin Variable Elimination]
	Let $\matr{a} \in \rationals^{m \times n}$, $\vect{b} \in \rationals^{m\times 1}$, and $j \in [n]$.
	Let further
    $m'=|\lbs{a}{j}|\cdot |\ubs{a}{j}|+|\nbs{a}{j}|$ and $\matr{f} \in \rationals^{m' \times m}$ be a matrix consisting of exactly the following rows:
    \begin{align*}
        a_{u,j}^{-1} \cdot \rvect{e}{u} - a_{\ell,j}^{-1} \cdot \rvect{e}{\ell}\ &\text{ for every pair }\ (\ell,u) \in \lbs{a}{j}\times\ubs{a}{j}\ \qquad \text{ and }\\
        \rvect{e}{i}\ &\text{ for every }\ i \in \nbs{a}{j}.
    \end{align*}

	\noindent Then the \emph{Fourier-Motzkin variable elimination} of $x_j$ from the system $\sys{a}{b}$ is defined as the system $\matr{f}\matr{a}\vect{x}\leq \matr{f} \vect{b}$.
\end{defi}

The FM method can also be used to check the consistency (i.e. satisfiability) of a system $\sys{a}{b}$, by successively eliminating all variables $x_n, \ldots, x_1$.
This yields intermediate systems $\matr{a}^{(n-1)}\vect{x}\leq \vect{b}^{(n-1)}, \ldots, \matr{a}^{(0)}\vect{x}\leq \vect{b}^{(0)}$, where the last one does not contain any variable, i.e. $\matr{A}^{(0)} = 0$.
Note that all entries of the transformation matrix $\matr{f}$ in the definition above are positive, and thus for any $k \in \{0,\ldots,n-1\}$ and any row $i'$ in $\matr{a}^{(k)}\vect{x}\leq \vect{b}^{(k)}$,
there exists $\vect{f} \in \rationals^{1\times m}$ such that $\vect{f} \geq 0$, $\vect{f}\matr{A} = \rvect{a}{i'}^{(k)}$ and $\vect{f}\vect{b} = b^{(k)}_{i'}$.%
We call this kind of linear combinations \emph{conical combinations}.
Now we know by Farkas' Lemma (\Cref{thm:farkas}) that, if $\matr{a}^{(0)}\vect{x}\leq \vect{b}^{(0)}$ is unsatisfiable, then so is $\sys{a}{b}$.
On the other hand, if $\matr{a}^{(0)}\vect{x}\leq \vect{b}^{(0)}$ is satisfiable, then it is satisfied by the empty assignment, which can be extended successively to a model of $\matr{a}^{(1)}\vect{x}\leq \vect{b}^{(1)},\ldots, \matr{a}^{(n-1)}\vect{x}\leq \vect{b}^{(n-1)}$ and $\sys{a}{b}$.

A major drawback of the Fourier-Motzkin variable elimination is its doubly exponential worst-case complexity in time and space w.r.t. the number of eliminated variables.
Moreover, many of the generated rows are redundant because they are conical combinations of the other rows, i.e. they could be omitted without changing the solution set of the system.

\begin{defi}[Redundancy]\label{def:redundancy}
	Let $\matr{A} \in \rationals^{m \times n}$ and $\vect{b} \in \rationals^{m\times 1}$. Let further $i \in [m]$ and $I := [m] \setminus \{i\}$.
	We call $\rvect{a}{i}\vect{x} \leq b_i$ \emph{redundant} in $\sys{a}{b}$, if $\solset{\sys{a}{b}} = \solset{\subsys{a}{b}{I}}$.
\end{defi}

Redundancies might already be contained in the input system, or they are introduced by the projection operation.
While removing all redundancies is expensive, there are efficient methods, e.g. Imbert's acceleration theorems \cite{imbert1990redundant,imbert1993fourier,JingComplexityFME}, for removing some redundancies of the latter type.
The following lemma describes this kind of redundancy. 

\begin{lem}[Redundancy by Construction]\label{def:redundancyconstruction}
  Let $\matr{A} \in \rationals^{m \times n}, \vect{b} \in \rationals^{m\times 1}$ and $\matr{F} \in \rationals^{m' \times m}$.
  Let furthermore $\matr{A'}=\matr{F}\matr{A}$, $\vect{b'}=\matr{F}\vect{b}$ and $i \in [m']$.
  If there exists $\vect{r} \in \rationals^{1\times m'}$ with $\vect{r} \geq 0$, $r_i = 0$ and $\vect{r}\matr{F} = \rvect{f}{i}$ (i.e. the $i$-th row of $\sys{a'}{b'}$ is a conical combination $\vect{r}\matr{F}\matr{A}\vect{x} \leq \vect{r}\matr{F}\vect{b}$ ($\vect{r}\matr{A'}\vect{x} \leq \vect{r}\vect{b'}$) of the other rows), then that row is redundant in $\sys{a'}{b'}$.
\end{lem}

\begin{exa}\label{example:redundancies}
	Consider the system on the left of \Cref{fig:redundancies} on which we applied Fourier-Motzkin to eliminate $x_1$ and $x_2$ to obtain the system on the right, whose first row is redundant by construction. The matrix $F$ from \Cref{def:redundancyconstruction} is encoded in the row labels.
    Note that we omit the eliminated variables and corresponding zero-columns to save space.
\end{exa}
\begin{figure}[h]
	\begin{center}
	\begin{tikzpicture}
		\node[draw, align=center] (root) at (0,0) {$ $\\ \vspace{-1.2em}
			$\begin{blockarray}{rccc}
				\begin{block}{r[ccc]}
					c_1 &  1 & 1 & 1 \bigstrut[t] \\
					c_2 &  1 & -1 & 1 \\
					c_3 &  -1 & 1 & 1 \\
					c_4 &  -1 & -1 &  1 \bigstrut[b]\\
				\end{block}
			\end{blockarray}
			\cdot
			\left[ {\begin{array}{cc}
				x_1 \\ x_2 \\ x_3
			\end{array} } \right] 
			\leq
			\begin{blockarray}{c}
				\begin{block}{[c]}
				0 \bigstrut[t] \\
				0 \\
				1 \\
				1 \bigstrut[b]\\
				\end{block}
			\end{blockarray}$
		};

		\node[draw, align=left] (A1) at (8,0) {$ $\\ \vspace{-1.2em}
			$\begin{blockarray}{rc}
				\begin{block}{r[c]}
					c_1+c_2+c_3+c_4 &  4 \bigstrut[t] \\
					c_1+c_4 &  2 \\
					c_2+c_3 &  2\bigstrut[b]\\
				\end{block}
			\end{blockarray}
			\cdot
			\left[ {\begin{array}{cc}
				x_3 
			\end{array} } \right] 
			\leq
			\begin{blockarray}{c}
				\begin{block}{[c]}
				2 \bigstrut[t] \\
				1 \\
				1 \bigstrut[b]\\
				\end{block}
			\end{blockarray}$
		};

		\draw[{Latex[length=2mm,width=2mm]}-] (A1) -- (root) node[pos=0.5,above=0.3cm, align=center] {elim. $x_1$ \\ and $x_2$};
	\end{tikzpicture}
	\end{center}
	\caption{Fourier-Motzkin elimination steps from \Cref{example:redundancies}.}\label{fig:redundancies}
\end{figure}

\section{FMplex as Variable Elimination Procedure}
\label{sec:varelim}

The FM method encodes that none of the lower bounds on some variable $x_j$ in a system $\sys{a}{b}$ is larger than any of its upper bounds.
In our \emph{FMplex} method, we do not consider all lower-upper bound combinations at once.
Instead, we \emph{split the problem into a set of sub-problems} by a case distinction either on \emph{which of the lower bounds are the largest} or, alternatively, on \emph{which of the upper bounds are the smallest}.
For the case distinction on lower bounds, we create one sub-problem for each lower bound on $x_j$, and in this sub-problem we consider solutions where that lower bound is maximal among all lower bounds, and at the same time not larger than any of the upper bounds.
The idea is that, if we find a solution for one of the sub-problems, it is also a solution of the original system.
On the other hand, any solution of $\sys{a}{b}$ will satisfy one of the sub-problems, since one of the lower bounds has to be maximal for that assignment.
The case distinction for upper bounds works analogously.

Asymptotically, these sub-problems are significantly smaller than the systems produced by FM, so that in total our approach produces \emph{at most exponentially} many constraints after iterated application, in contrast to the doubly exponential effort of the FM method.

Formally, if there are no upper or no lower bounds on $x_j$, then there is no need for case  splitting and we follow FM using $\exists x_j.\; \sys{a}{b} \equiv \subsys{a}{b}{\nbs{a}{j}}$.
Otherwise, for the sub-problem when designating $i \in \lbs{a}{j}$ as a largest lower bound, we encode that no other lower bound is larger than the bound induced by row $i$, and no upper bound is below this bound.
Using the set notation for systems, we obtain
\begin{align*}
  &\{ \bnd{j}{\rvect{a}{i'}\;\vect{x} \leq b_{i'}} \leq \bnd{j}{\rvect{a}{i}\;\vect{x} \leq b_i} \mid i' \in \lbs{a}{j} \setminus \{ i \} \} \\
	\cup\ & \{\bnd{j}{\rvect{a}{i}\;\vect{x} \leq b_i} \leq \bnd{j}{\rvect{a}{i'}\;\vect{x} \leq b_{i'}} \mid i' \in \ubs{a}{j}\} %
	\cup  \{\rvect{a}{i'}\;\vect{x} \leq b_{i'} \mid i' \in \nbs{a}{j}\}.
\end{align*}

\begin{exa} \label{example:fmplex-idea}
	We eliminate $x_2$ from the system $\sys{a}{b}$ consisting of the lower-bounding constraints $c_1$ and $c_2$, and the upper-bounding $c_3$ and $c_4$, specified below along with a graphical depiction.
	
	\begin{figure}[h]
		\begin{subfigure}[h]{0.48\textwidth}
			\centering
			\begin{align*}
				\begin{blockarray}{ccc}
					\begin{block}{c[cc]}
						c_1 & -1 & -1 \bigstrut[t] \\
						c_2 & 0 & -2 \\
						c_3 & -2 & 1 \\
						c_4 & 0 & 1 \bigstrut[b]\\
					\end{block}
				\end{blockarray}
				\cdot
				\left[ {\begin{array}{ccc}
					x_1 \\ x_2
				\end{array} } \right]
				\leq
				\left[ {\begin{array}{ccc}
					-4 \\ -2 \\ 1 \\ 5
				\end{array} } \right]
			\end{align*}
		\end{subfigure}\hfill
		\begin{subfigure}[h]{0.48\textwidth}
			\centering
			\begin{tikzpicture}[scale=0.5,
				lb/.style={
					very thick, color=blue100, opacity=1, shorten <= -5pt, shorten >= -5pt,
					postaction={draw,decorate,decoration={markings,mark=at position .4 with {\draw[->, shorten <= 0, shorten >= 0] (0,0) -- (0,-8pt);}}}}, 
					ub/.style={
						very thick, color=green100, opacity=1, shorten <= -5pt, shorten >= -5pt,
						postaction={draw,decorate,decoration={markings,mark=at position .6 with {\draw[->, shorten <= 0, shorten >= 0] (0,0) -- (0,8pt);}}}},
						]
						\draw[->] (0,0) -- (6.5,0) node[right] {$x_1$};
						\draw[->] (0,0) -- (0,6) node[below left] {$x_2$};
						\path[name path = c1] (0,4) -- (4,0);
						\path[name path = c2] (0,1) -- (6,1);
						\path[name path = c3] (0,1) -- (2.5,6);
						\path[name path = c4] (0,5) -- (6,5);
						\path[name path = c5]  (6,0) -- (6,6);
						\path [name intersections={of=c1 and c3, by={p13}}];
						\path [name intersections={of=c3 and c4, by={p34}}];
						\path [name intersections={of=c4 and c5, by={p45}}];
						\path [name intersections={of=c2 and c5, by={p25}}];
						\path [name intersections={of=c1 and c2, by={p12}}];
						\fill[color=gray!25] (p13) -- (p34) -- (p45) -- (p25) -- (p12) -- cycle;
						\draw[lb] (0,4) -- (4,0) node[left, pos=0, xshift=1pt] {$c_1$};
						\draw[lb] (0,1) -- (6,1) node[above right, yshift=-2pt] {$c_2$};
						
						\draw[ub] (0,1) -- (2.25,5.5) node[left, xshift=2pt, yshift=1pt] {$c_3$};
						\draw[ub] (0,5) -- (6,5) node[below right] {$c_4$};
						
						\draw[color=black75, thick, dashed] (3,0) -- (3,6);
					\end{tikzpicture}
				\end{subfigure}	
			\end{figure}
			
			\noindent The lower bounds $\lbs{a}{2} = \{1,2\}$ on $x_2$ are blue, the upper bounds $\ubs{a}{2} = \{3,4\}$ are green.
			The solution set is the gray area and the dashed line indicates the split into two sub-problems, namely the cases that $c_1$ resp. $c_2$ is a largest lower bound on $x_2$ and not larger than any upper bound on $x_2$.
			The encoding of the $c_1$-case is given by 
			\[\bnd{2}{c_2}\leq \bnd{2}{c_1}\quad\land\quad \bnd{2}{c_1} \leq \bnd{2}{c_3}\quad \land\quad \bnd{2}{c_1} \leq \bnd{2}{c_4},\]
			where the first constraint expresses that $c_1$ gives the greatest lower bound and the other two ensure that it does not exceed the upper bounds.
			In normal form, this formula evaluates to $(x_1 \leq 3) \land (-3x_1 \leq -3) \land (-x_1 \leq 1)$ satisfied by any $x_1 \in [1,3]$, on the left of the dashed line.
			The case for $c_2$ evaluates to $(-x_1 \leq -3) \land (-2x_1 \leq 0) \land (0 \leq 4)$ and is satisfiable on the right of the dashed line.
			The disjunction of the two formulas then defines exactly those values for $x_1$ which allow a solution of the initial system, i.e. $x_1 \in [1,3]$ or $x_1 \in [3,\infty)$.
		\end{exa}
		
		The construction for the case $i \in \ubs{a}{j}$ designating $i$ as smallest upper bound is analogous.
		Like for FM, the constructed constraints $\bnd{j}{\rvect{a}{i'}\;\vect{x} \leq b_{i'}} \leq \bnd{j}{\rvect{a}{i}\;\vect{x} \leq b_i}$ are equivalent to linear combinations of the system $\sys{a}{b}$.
		However, combinations of two lower bounds or two upper bounds use negative coefficients, meaning that these linear combinations are not necessarily \emph{conical}.
		Nevertheless, we can formulate the sub-problems with respect to the matrix representation, resulting in the following transformation:

\begin{defi}[Restricted Projection] \label{def:restricted-projection}
	Let $\matr{A} \in \rationals^{m \times n}$, $\vect{b} \in \rationals^{m\times 1}$ and $j \in [n]$.
	\begin{itemize}
		\item If $\lbs{a}{j}\not=\emptyset$ and $\ubs{a}{j}\not=\emptyset$, then for any $i \in \lbs{a}{j} \cup \ubs{a}{j}$ we fix $\matr{f} \in \rationals^{(m-1) \times m}$ arbitrarily but deterministically to consist of exactly the following rows:
		\begin{align*}
			\frac{1}{a_{i,j}} \cdot \rvect{e}{i} - \frac{1}{a_{i',j}} \cdot \rvect{e}{i'} &\text{ for every } i' \in \lbs{a}{j} \setminus \{i\},\\
			-\frac{1}{a_{i,j}} \cdot \rvect{e}{i} + \frac{1}{a_{i',j}} \cdot \rvect{e}{i'} &\text{ for every } i' \in \ubs{a}{j} \setminus \{i\},\qquad \text{ and }\qquad \rvect{e}{i'} \text{ for every } i' \in \nbs{a}{j}.
		\end{align*} 
		Then the \emph{restricted projection} $\fmpproj{\sys{a}{b}}{j}{i}$ of $x_j$ from the system $\sys{a}{b}$ with respect to the row $i$ is defined as the system $\matr{f}\matr{a}\vect{x}\leq \matr{f} \vect{b}$.
		We call $\matr{f}$ the \emph{projection matrix} corresponding to $\fmpproj{\sys{a}{b}}{j}{i}$.
		\item If $\lbs{a}{j}=\emptyset$ or $\ubs{a}{j}=\emptyset$, then we define the projection matrix $\matr{f} \in \rationals^{|\nbs{a}{j}| \times m}$ to have exactly one row $\rvect{e}{i'}$ for each $i' \in \nbs{a}{j}$, and define $\fmpproj{\sys{a}{b}}{j}{\bot}$ as $\matr{f}\matr{a}\vect{x}\leq \matr{f} \vect{b}$.
	\end{itemize}
\end{defi}

We call this a restricted projection because it defines exactly the part of the projection, where the respective bound is a greatest lower or least upper bound.
Crucially, the solutions of the restricted projections for all lower (or all upper) bounds of a variable exactly cover the projection of the entire solution set, as is demonstrated by the following lemma.

\begin{lem} \label{lma:restricted-projection}
	Let $\matr{A} \in \rationals^{m \times n}$, $\vect{b} \in \rationals^{m\times 1}$, $j \in [n]$.
	If $\lbs{a}{j}\neq \emptyset$ and $\ubs{a}{j}\neq \emptyset$, then
	\[
		\proj{\solset{\sys{a}{b}}}{X\setminus\{x_j\}} = \bigcup_{i \in \lbs{a}{j}} \solset{\fmpproj{\sys{a}{b}}{j}{i}} = \bigcup_{i \in \ubs{a}{j}} \solset{\fmpproj{\sys{a}{b}}{j}{i}}.
	\]
	Otherwise $(\lbs{a}{j}=\emptyset$ or $\ubs{a}{j}=\emptyset)$, it holds $\proj{\solset{\sys{a}{b}}}{X\setminus\{x_j\}} = \solset{\fmpproj{\sys{a}{b}}{j}{\bot}}$.
\end{lem}
\begin{proof}
	The case $\lbs{a}{j}=\emptyset$ or $\ubs{a}{j}=\emptyset$ follows from the correctness of FM.
	We show the equality for $\lbs{a}{j}$, the case for $\ubs{a}{j}$ is analogous.
	\begin{description}
		\item[$\supseteq$] Let $i \in \lbs{a}{j}$ and $\alpha \models \fmpproj{\sys{a}{b}}{j}{i}$, then for all $\ell \in \lbs{a}{j}$, $u \in \ubs{a}{j}$ it holds 
		\[\alpha(\bnd{j}{\rvect{a}{\ell}\;\vect{x} \leq b_\ell}) \leq \alpha(\bnd{j}{\rvect{a}{i}\;\vect{x} \leq b_i}) \leq \alpha(\bnd{j}{\rvect{a}{u}\;\vect{x} \leq b_u}).\]
Thus, $\alpha[x_j \mapsto \alpha(\bnd{j}{\rvect{a}{i}\;\vect{x} \leq b_i})] \models \sys{a}{b}$.

		\item[$\subseteq$] Let $\alpha \models \sys{a}{b}$. Then, by Fourier-Motzkin, for all $\ell \in \lbs{A}{j}$ and $u \in \ubs{A}{j}$ it holds
		$\alpha(\bnd{j}{\rvect{a}{\ell}\;\vect{x} \leq b_\ell}) \leq \alpha(\bnd{j}{\rvect{a}{u}\;\vect{x} \leq b_u})$.
		Let $i \in \lbs{A}{j}$ so that
		\[
			\alpha(\bnd{j}{\rvect{a}{i}\;\vect{x} \leq b_i}) = \max\{\alpha(\bnd{j}{\rvect{a}{\ell}\;\vect{x} \leq b_\ell}) \mid \ell \in \lbs{a}{j} \},
		\]
		then $\alpha \models \fmpproj{\sys{a}{b}}{j}{i}$.\qedhere
	\end{description}
\end{proof}

\begin{defi}[FMplex Variable Elimination]
	For $\matr{A} \in \rationals^{m \times n}$, $\vect{b} \in \rationals^{m\times 1}$, $j \in [n]$ and $* \in \{ -,+ \}$, we define
	\begin{align*}
		\fmplexelim^*_{j}(\sys{a}{b}) = & \left.\begin{cases}
		\{\fmpproj{\sys{a}{b}}{j}{i} \mid i \in \sbs{a}{j}\} & \text{ if } \lbs{a}{j} \neq \emptyset \textit{ and }  \ubs{a}{j}\neq\emptyset\\
		\{ \fmpproj{\sys{a}{b}}{j}{\bot}  \} &  \textit{otherwise}.
		\end{cases}\right.
	\end{align*}
\end{defi}

The FMplex elimination defines a set of restricted projections which can be composed to the full projection according to \Cref{lma:restricted-projection}.
Lifting this from sets to logic naturally results in the following theorem which demonstrates the usage of our method.

\begin{thm}\label{thm:correctness-elim}
	Let $\matr{A} \in \rationals^{m \times n}$, $\vect{b} \in \rationals^{m\times 1}$, and $j \in [n]$. Then
	\[\exists x_j.\; \sys{a}{b}\quad \equiv \quad\bigvee_{S \in \fmplexelim^-_{j}(\sys{a}{b}) } S  \quad\equiv\quad \bigvee_{S \in \fmplexelim^+_{j}(\sys{a}{b})} S.\] 
\end{thm}
\begin{proof}
	This immediately follows from \Cref{lma:restricted-projection}.
\end{proof}

For eliminating multiple variables, we iteratively apply $\fmplexelim^-$ or $\fmplexelim^+$ to each restricted projection resulting from the previous elimination step.
Note that we can choose the next variable to be eliminated as well as the bound type independently in every branch.

\begin{exa}\label{example:fmplex-elim-procedure}
	We continue \Cref{example:fmplex-idea}, where we eliminated $x_2$, and next eliminate $x_1$: 
	\begin{align*}
		\exists x_1. \; \exists x_2.\; \sys{a}{b}\quad \equiv\quad \exists x_1.\bigvee_{S\in \fmplexelim^-_{2}(\sys{a}{b})} S \quad \equiv& \quad
		\exists x_1.\; \left(x_1 \leq 3 \land -3x_1 \leq -3\land -x_1 \leq 1\right)\;\\[-10pt]
		&\;\lor  \exists x_1. \;\left(-x_1 \leq -3 \land -2 x_1 \leq 0 \land 0 \leq 4\right) 
	\end{align*}
	We eliminate the two quantifiers for $x_1$ separately, using
	\begin{align*}
		\fmplexelim^-_{1}\left(\;x_1 \leq 3 \land -3x_1 \leq -3\land -x_1 \leq 1\;\right) &= \{\;(-1 \leq 1 \land 1 \leq 3),\; (1 \leq -1\land -1 \leq 3 )\;\}\\
		\text{and } \fmplexelim^-_{1}\left(\; -x_1 \leq -3 \land -2 x_1 \leq 0 \land\ 0 \leq 4\;\right) &= \{\;(0 \leq 4)\;\}, \text{ giving us the final result}
	\end{align*}
	\[
		\exists x_1. \; \exists x_2.\; \sys{a}{b}\quad \equiv\quad \big(\;(-1 \leq 1 \land 1 \leq 3)\quad \lor\quad (1 \leq -1\land -1 \leq 3)\;\big) \quad\lor\quad (0 \leq 4).
	\]
	Note that in any of the systems, we could use $\fmplexelim^+$ instead of $\fmplexelim^-$.
	For example, the first disjunct obtained from eliminating $x_2$ contains two lower bounds on $x_1$, but only one upper bound.
	Thus, one can reduce the number of resulting disjuncts, by instead computing
	\[\fmplexelim^+_{1}\left(\;x_1 \leq 3 \land -3x_1 \leq -3\land -x_1 \leq 1\;\right) = \{(1 \leq 3 \land -1 \leq 3)\}.\]
\end{exa}
	
We analyze the complexity in terms of the number of new rows (or constraints) that are constructed during the elimination of all variables:

\begin{thm}[Complexity of \fmplexelim]
	Let $\matr{A} \in \rationals^{m \times n}$, and $\vect{b} \in \rationals^{m\times 1}$.
	When eliminating $n$ variables from $\sys{a}{b}$ by repeated application of $\fmplexelim^-$ or $\fmplexelim^+$, at most $\mathcal{O}(n\cdot m^{n+1})$ new rows are constructed.
\end{thm}
\begin{proof}
	Let $N(m,n)$ be the maximal number of constructed rows when eliminating $n$ variables from a system of $m$ rows.
	Applying $\fmplexelim^-$ or $\fmplexelim^+$ once leads to at most $m-1$ restricted projections, each of which contains at most $m-1$ rows and $n-1$ variables.
	This implies $N(m,n) \leq (m-1)\cdot((m-1) + N(m-1,n-1))$, and repeating this reasoning yields
	\[N(m,n) \leq \sum\limits_{i=1}^{k}(m-i)\cdot \prod\limits_{j=1}^{i}(m-j) \leq n\cdot m^{n+1}, \text{ where } k = \min(n,m).\qedhere\]
\end{proof}

The following example illustrates that this exponential complexity can occur in practice, thus the bound is strong.

\begin{exa}\label{example:worstcase}
	We use $\fmplexelim^-$ to eliminate $x_1, \ldots, x_{n}$ in that order from the constraint set 
	\[
		\{\;-x_j - x_{n+1} \leq 0 \mid j \in [n]\;\}\; \cup\; \{\;-x_j - 2x_{n+1} \leq 0 \mid j \in [n]\;\}\; \cup\; \{\;x_1 + \ldots + x_{n+1} \leq -1\;\}.
	\]
	Any intermediate system where the first $k \in \{0,...,n\}$ variables have been eliminated ($k=0$ gives the initial system) has the following form:
	There are two lower bounds $-x_j - x_{n+1} \leq 0$ and $-x_j - 2x_{n+1} \leq 0$ on each of the remaining variables $x_{k+1},\ldots,x_{n}$, and there is one constraint $x_{k+1} + \ldots + x_{n} + ax_{n+1} \leq -1$ which is the only upper bound on all those variables.
	This upper bound is the result of combining $(x_1 + \ldots + x_{n+1} \leq -1)$ with the chosen lower bounds and thus, $a < 0$ holds if $k > 1$.
	Moreover, there are constraints $(x_{n+1} \leq 0)$ and/or $(-x_{n+1} \leq 0)$ resulting from combining the two lower bounds.

	Accordingly, each intermediate system has two branches for the next elimination, each of which requires the computation of two constraints.
	Thus, the total number of generated constraints is $2\cdot (2 + 4 + \ldots + 2^{n}) = 2\cdot(2^{n+1}-2)$.

	Interestingly, every leaf of this elimination contains a constraint $ax_{n+1} \leq -1$, with $a < 0$, i.e. a lower bound on $x_{n+1}$ that implies $x_{n+1} > 0$.
	If any of the greatest lower bounds chosen to create a leaf had the form $-x_j - 2x_{n+1} \leq 0$, then that resulting leaf contains $x_{n+1} \leq 0$ and is thus unsatisfiable.
	Consequently, the only satisfiable leaf is obtained by choosing the lower bound $-x_j-x_{n+1} \leq 0$ for all $j \in \{ 0,\ldots,k \}$.
\end{exa}

While still exponential, this bound is considerably better than the theoretical doubly exponential worst-case complexity of the FM method.
Shortly speaking, FMplex trades one exponential step at the cost of the result being a decomposition into multiple partial projections.
However, there are systems for which FMplex produces strictly more rows than the FM method: 
assume a system of $m$ constraints, with each $m/2$ upper and lower bounds on a variable $x_1$, then $\fmplexelim^-$ (or $\fmplexelim^+$) constructs $(m/2)\cdot (m-1)$ new constraints in the elimination of $x_1$, while FM would only compute $(m/2)^2$.
Here, the advantage of FMplex becomes apparent only for the elimination of multiple variables, when the split into sub-problems prevents the doubly exponential growth.

Like FM, FMplex keeps redundancies from the input throughout the algorithm, thus there might be identical rows in the same or across different sub-problems.
But in contrast to FM, FMplex does not introduce any redundancies by construction in the sense of \Cref{def:redundancyconstruction}.

\begin{thm}
	Let $\matr{A} \in \rationals^{m \times n}$, $\vect{b} \in \rationals^{m\times 1}$ and $k\in[m]$.
	Assume $(\matr{a}^{(0)}\vect{x} \leq \vect{b}^{(0)}) = (\sys{a}{b})$ and for all $j \in [k]$, let $(\matr{a}^{(j)}\vect{x} \leq \vect{b}^{(j)}) \in \fmplexelim^-_{j}(\matr{a}^{(j-1)}\vect{x} \leq \vect{b}^{(j-1)}) \cup \fmplexelim^+_{j}(\matr{a}^{(j-1)}\vect{x} \leq \vect{b}^{(j-1)})$.
Let $\matr{F^{(1)}}, \ldots, \matr{F^{(k)}}$ be the respective projection matrices, and $\matr{F} = \matr{F^{(k)}} \cdot \ldots \cdot \matr{F^{(1)}}$.
	Then $F$ is linearly independent.
\end{thm}
\begin{proof}
  By definition, $\matr{F^{(k)}}$, $\ldots$, $\matr{F^{(1)}}$ are linearly independent.
  Thus, $\matr{F}$ is also linearly independent as a product of linearly independent matrices.
\end{proof}

\section{FMplex as Satisfiability Checking Procedure}
\label{sec:sat}

A formula is satisfiable if and only if eliminating all variables (using any quantifier elimination method such as FM or FMplex) yields a tautology.
However, FMplex applies case splitting such that satisfiability of any of the cases implies the satisfiability of the original problem.
Therefore, we do not need to compute the whole projection at once, but we can explore the decomposition using a depth-first search.
In this section, we first present a basic version of our algorithm to exploit this observation. Later, we examine how the search tree can be pruned, resulting in two additional variants.
The resulting search tree has the original system as root, and each node has as children the systems resulting from restricted projections.
The original system is satisfiable if and only if there is a leaf whose constraints are all tautologies.
An example is depicted in \Cref{figure:example-searchtree}.

\begin{figure}[b]
	\centering
	\begin{tikzpicture}
		\tikzstyle{level 1} = [level distance=15mm, sibling distance = 50mm]
		\tikzstyle{level 2} = [level distance=15mm, sibling distance = 40mm]
			\node[] (root) {$\sys{a}{b}$} 
				child{node[] (v11) {$\fmpproj{\sys{a}{b}}{2}{1}$}
					child{node[] (v21) {$(0 \leq 2)\land (0 \leq 2)$}}
					child{node[] (v22) {$(0 \leq 4) \land (0 \leq -2)$}}
				}
				child{node[] (v12) {$\fmpproj{\sys{a}{b}}{2}{2}$}
					child{node[] (v23) {$(0 \leq 4)$}}
				};
		\end{tikzpicture}
	
	\caption{%
	The search tree corresponding to \Cref{example:fmplex-elim-procedure}.
	The very first leaf (bottom left) is already satisfiable, meaning that the rest would not need to be computed.%
	}\label{figure:example-searchtree}
\end{figure}

An important observation is that we can decide independently for each node of the search tree which variable to eliminate next and whether to branch on lower or on upper bounds.
This freedom of choice is formalized in the following definition.

\begin{defi}[Branch Choices]
	The set of \emph{branch choices} for a system $\sys{a}{b}$ is
	\begin{align*}
		\textit{branch\_choices}(\sys{a}{b})\ =&\ \ \{\; (x_j, I^{*}_j(A)) \mid * \in \{-,+\} \;\wedge\; j \in [n]\;\wedge\; \lbs{a}{j} \neq \emptyset \neq  \ubs{a}{j}\;\} \\
		&\ \cup\ \{\; (x_j, \{\bot\}) \mid j \in [n]\;\wedge\; (\lbs{a}{j} = \emptyset \vee \ubs{a}{j} = \emptyset) \;\}.
	\end{align*}
\end{defi}

For an initial input $\widehat{\matr{a}}\vect{x} \leq \widehat{\vect{b}}$ with $\widehat{m}$ rows, we define the depth-first search using the recursive method $\texttt{FMplex}(\widehat{\matr{a}}\vect{x} \leq \widehat{\vect{b}};\sys{a}{b},\matr{f})$ in \Cref{algo:base} where $\sys{a}{b}$ is the currently processed sub-problem in the recursion tree. The original system $\widehat{\matr{a}}\vect{x} \leq \widehat{\vect{b}}$ with $\widehat{m}$ is not used in the algorithm itself, but included for illustrational purposes.
We track the relation of $\sys{a}{b}$ to $\widehat{\matr{a}}\vect{x} \leq \widehat{\vect{b}}$ in terms of linear combinations using the parameter $\matr{f}$.
The initial call is defined as $\texttt{FMplex}(\widehat{\matr{a}}\vect{x} \leq \widehat{\vect{b}}) = \texttt{FMplex}(\widehat{\matr{a}}\vect{x} \leq \widehat{\vect{b}};\widehat{\matr{a}}\vect{x} \leq \widehat{\vect{b}},\matr{e})$.
We allow that $\sys{a}{b}$ contains identical rows when they are obtained in different ways (which is reflected by $\matr{f}$).
We need to keep these duplicates for proving the results of this section.

\begin{algorithm}[b]%
	\linespread{1.25}\selectfont
    \caption{$\texttt{FMplex}(\widehat{\matr{a}}\vect{x} \leq \widehat{\vect{b}};\sys{a}{b},\matr{f})$}\label{algo:base}
	\SetKwInOut{Data}{Data}
    \SetKwInOut{Input}{Input}
    \SetKwInOut{Output}{Output}
    \DontPrintSemicolon
    \SetKwFunction{FCheck}{FMplex}
    \SetKwProg{Fn}{Function}{:}{}
	\Data{$\widehat{\matr{a}} \in \rationals^{\widehat{m} \times n},\ \widehat{\vect{b}}\in \rationals^{\widehat{m}}$}
    \Input{$\matr{a} \in \rationals^{m \times n},\ \vect{b}\in \rationals^{m}$, $\matr{f} \in \rationals^{m \times \widehat{m}}$ s.t. $\matr{f}\widehat{\matr{a}} = \matr{a}$ and $\matr{f}\widehat{\vect{b}} = \vect{b}$}
    \Output{(\sat, $\alpha$) with $\alpha \models \matr{a}\vect{x} \leq \vect{b}$, or (\unsat, $K$) where the rows $K \subseteq [\widehat{m}]$ of $\widehat{\matr{a}}\vect{x} \leq \widehat{\vect{b}}$ build a minimal unsatisfiable subset of constraints $\{ \vect{\widehat{a}_{k,-}} \vect{x} \leq \widehat{b_k} \mid k \in K \}$, or
	\partialunsat\\\hspace*{-4em}\hrulefill}
	\lIf{$\matr{a} = 0 \land \vect{b} \geq 0$} {
		\Return{$($\sat, $\alpha: \emptyset \to \reals)$\tcp*[f]{trivial sat}} 
	}
    \If(\tcp*[f]{global conflict}){$\exists i \in [m].\; \rvect{a}{i} = 0 \land b_i < 0 \land \rvect{f}{i} \geq 0$}{
            \Return{$(\unsat, \{ i' \mid f_{i,i'} \neq 0 \})$} 
    }
	\textbf{choose} $(x_j,V) \in \textit{branch\_choices}(\matr{a}\vect{x} \leq \vect{b})$\;
    \ForEach{$i \in V$}{
		\textbf{compute} $\matr{a'}\vect{x} \leq \vect{b'} := \fmpproj{\matr{a}\vect{x} \leq \vect{b}}{j}{i}$ with projection matrix $\matr{f'}$ \;
        \Switch{\FCheck($\sys{\widehat{a}}{\widehat{b}};\matr{a'}\vect{x} \leq \vect{b'},\matr{f'}\cdot\matr{f}$)}{
			\lCase{$(\unsat, K)$}{\Return{$(\unsat, K)$}}
        	\lCase{$(\sat, \alpha)$}{\Return{$(\sat, \alpha[x_j \mapsto r])$} for a suitable $r \in \rationals$, e.g. $r = \bnd{j}{\vect{a_{i,-}}\vect{x}\leq b_i}$}
			\lCase{\partialunsat}{\textbf{continue}}
    	}
    }
	\Return{\partialunsat}
\end{algorithm}

\paragraph{Solutions} If a trivially satisfiable node is found, the algorithm constructs an assignment starting with the empty assignment and extending it with values for all eliminated variables in reverse elimination order. 
For every variable $x_j$, the assignment $\alpha$ is extended by assigning to $x_j$ a value that us not below any lower bound and not above any upper bound on $x_j$, when the bounds are evaluated under $\alpha$.
By the semantics of the projection, the value of the designated (largest lower or smallest upper) bound on $x_j$ is suitable.

\paragraph{Conflicts} We distinguish inconsistencies in $\sys{a}{b}$ by the following notions: 
We call a row $i$ of $\sys{a}{b}$ a \emph{conflict} if it is of the form $\rvect{a}{i} = \vect{0}$ with $b_i < 0$.
We call the conflict \emph{global} if $\rvect{f}{i} \geq 0$  and \emph{local} otherwise.
In case of a global conflict, Farkas' Lemma allows to deduce the unsatisfiability of $\sys{\widehat{a}}{\widehat{b}}$, thus stopping the search before the whole search tree is generated.
Then a set of conflicting rows $K$ of the input system corresponding to $\rvect{f}{i}$ is returned.
In particular, the set $\{\rvect{\widehat{a}}{j}\; \vect{x} \leq \widehat{b}_j \mid f_{i,j} \neq 0\}$ is a minimal unsatisfiable subset of the constraints in $\sys{\widehat{a}}{\widehat{b}}$.
In case of a local conflict, the current node is unsatisfiable, but we need to continue the search in the other branches.
The algorithm returns $\partialunsat$ to indicate that $\sys{a}{b}$ is unsatisfiable, but the unsatisfiability of $\sys{\widehat{a}}{\widehat{b}}$ cannot be derived.
This approach, formalized in \Cref{algo:base}, guarantees that the initial call will never return $\partialunsat$; we always find either a global conflict or a solution.

\vspace{0.75em}
\noindent The correctness and completeness of \texttt{FMplex} follows from \Cref{thm:correctness-elim} and \Cref{thm:global-conflict}.
As an immediate consequence of \Cref{thm:correctness-elim}, \texttt{FMplex} returns \sat~ if and only if the input is satisfiable.
On the other hand, if the input is unsatisfiable, then the initial call to \Cref{algo:base} will always return \unsat~ and never \partialunsat, as is shown in \Cref{thm:global-conflict}.

\begin{thm} \label{thm:global-conflict}
	Let $\matr{\widehat{a}} \in \rationals^{\widehat{m} \times n}$, and $\vect{\widehat{b}} \in \rationals^{\widehat{m}\times 1}$.
	Then $\sys{\widehat{a}}{\widehat{b}}$ is unsatisfiable if and only if the call $\texttt{FMplex}(\sys{\widehat{a}}{\widehat{b}})$ to \Cref{algo:base} terminates with a global conflict.
\end{thm}
\begin{proof}
	We show the two directions of the equivalence separately:

\vspace{0.5em}
\noindent$\Leftarrow$:\quad Assume the call $\texttt{FMplex}(\sys{\widehat{a}}{\widehat{b}};\matr{{a}}\vect{x} \leq \vect{{b}},F)$ in the recursion tree returns \unsat, i.e. $\sys{a}{b}$ contains a global conflict in some row $i$.
Then $\rvect{a}{i} = \rvect{f}{i} \matr{\widehat{a}} = 0$, $b_i = \rvect{f}{i} \cdot \vect{\widehat{b}} < 0$ and $\rvect{f}{i} \geq 0$, thus Farkas' Lemma (\Cref{thm:farkas}) implies that $\sys{\widehat{a}}{\widehat{b}}$ is unsatisfiable.

\vspace{0.5em}
\noindent$\Rightarrow$:\quad Assume $\sys{\widehat{a}}{\widehat{b}}$ is unsatisfiable.
Then there is a minimal subset $\widehat{K} \subseteq [\widehat{m}]$ with size $|\widehat{K}| \leq n + 1$ such that the corresponding rows $\{ \vect{\widehat{a}_{i,-}} \vect{x} \leq \widehat{b}_i \mid i \in \widehat{K} \}$ are together unsatisfiable.
We show that the algorithm will eventually construct a conflict from such a minimal unsatisfiable subset.
Then, we prove that this will be a global conflict, which will prompt the algorithm to return \unsat.
				
Let $(x_j,V) \in \textit{branch\_choices}(\sys{\widehat{a}}{\widehat{b}})$ and $i^*\in V$.
If $\widehat{K}\subseteq\nbs{\widehat{a}}{j}$, then each of the rows in $\widehat{K}$ are contained also in $\fmpproj{\sys{\widehat{a}}{\widehat{b}}}{j}{i^*}$.
Note that since $|\widehat{K}|>0$, this case cannot always apply until termination (except for the degenerate case where a row contains only zeros).
	
Otherwise, both $\lbs{\widehat{a}}{j} \cap \widehat{K} \neq \emptyset$ and $\ubs{\widehat{a}}{j} \cap \widehat{K}  \neq \emptyset$ must hold due to the minimality of $\widehat{K}$ (if only one intersection was non-empty, we could remove it from $\widehat{K}$ to still obtain an infeasible subset). Thus, at some point the algorithm will choose  $i^*=i$ for some $i \in \widehat{K}$.

Let $(\sys{a}{b})=\fmpproj{\sys{\widehat{a}}{\widehat{b}}}{j}{i}$ be the resulting partial projection.
Since $\sys{\widehat{a}}{\widehat{b}}$ is unsatisfiable, also $\sys{a}{b}$ is unsatisfiable by \Cref{lma:restricted-projection}.
With the same reasoning, $\matr{a'}\vect{x}\leq\vect{b'} := \fmpproj{\subsys{\widehat{a}}{\widehat{b}}{\widehat{K}}}{j}{i}$ is unsatisfiable.
We observe that all rows of $\sys{a'}{b'}$ are contained in $\sys{a}{b}$, thus there exists a minimal set $K$ of rows of $\sys{a}{b}$ that are also contained in $\sys{a'}{b'}$ such that the corresponding rows are a minimal unsatisfiable subset.
In particular, we have $0<\abs{K} < |\widehat{K}|$, because $\sys{a'}{b'}$ has less than $|\widehat{K}|$ rows.
	
We can repeat this reasoning for the successive elimination of all variables to obtain a system $\sys{a}{b}$ where the corresponding minimal unsatisfiable set has size $1$.
By construction, there exists a matrix $\matr{f}$ such that $\matr{a}=\matr{f}\matr{\widehat{a}}$ and $\vect{b}=\matr{f}\vect{\widehat{b}}$.
This means that there is a conflicting row $i$ in $\sys{a}{b}$ which is a linear combination of rows from $\subsys{\widehat{a}}{\widehat{b}}{\widehat{K}}$.
	
It remains to show that  $\rvect{f}{i} \geq 0$.
By construction, $f_{i,i} > 0$ holds, and for all $i' \in [m] \setminus \widehat{K}$ holds $f_{i,i'} = 0$.
Towards a contradiction, assume that there was some $i' \in \widehat{K}$ with $f_{i,i'} < 0$.
Since $\widehat{K}$ induces a minimal unsatisfiable subset, by Farkas' Lemma there is some other vector $\vect{f'} \in \rationals^{m\times 1}$ with $\vect{f'} \geq 0$ such that $(f'_i \neq 0 \Leftrightarrow i \in \widehat{K})$, $\vect{f'}\matr{\widehat{a}} = 0$ and $\vect{f'}\vect{\widehat{b}} < 0$.
Now let 
\[\lambda := \max\{\frac{-f_{i,j}}{f'_j}\mid f_{i,j} < 0\}\qquad \text{ and }\qquad \vect{g} := \rvect{f}{i} + \lambda \vect{f'}.\]
We get $\vect{g} \matr{\widehat{a}} = 0$ and $\vect{g}\vect{\widehat{b}} < 0$ because $\lambda > 0$.
But now, $\vect{g} \geq 0$, $g_{i'} = 0$ for all $i' \in [\widehat{m}] \setminus \widehat{K}$ and $g_{i'} = 0$ for some $i' \in \widehat{K}$, which implies that $\widehat{K}$ does not induce a minimal unsatisfiable subset.
As this is a contradiction, $\rvect{f}{i} \geq 0$ must already hold.
\end{proof}

In the next two subsections, we look at ways to prune the recursion tree, eventually leading to \Cref{algo:combined}. Additionally, the first improvement is based on insights that reveal a connection to the simplex method.

\subsection{Avoiding Redundant Checks}\label{subsec:avoiding-redundancies}

We observe that each row in a sub-problem $\sys{a}{b}$ in the recursion tree of $\texttt{FMplex}(\sys{\widehat{a}}{\widehat{b}})$ corresponds to a unique row $\hat{\imath}$ in $\sys{\widehat{a}}{\widehat{b}}$ in the sense that it is a linear combination of the rows $\{ \hat{\imath} \} \cup \mathcal{N}$ of $\sys{\widehat{a}}{\widehat{b}}$, where $\mathcal{N} \subseteq [\widehat{m}]$ corresponds to the lower/upper bounds designated as largest/smallest one to compute $\sys{a}{b}$.

The intuition is simple for a single elimination: for a restricted projection $\fmpproj{\sys{\widehat{a}}{\widehat{b}}}{j}{i}$, we have $\mathcal{N} = \{i\}$ containing only the designated bound.
Each row $\fmpproj{\sys{\widehat{a}}{\widehat{b}}}{j}{i}$ is a linear combination of the row $i$ and a unique row different from $i$.
This gives a one-to-one correspondence, which we can use to identify the next row that is added to $\mathcal{N}$ in the subsequent elimination.

\begin{thm} \label{thm:basis-nonbasis-correspondence}
	Let $\matr{\widehat{a}} \in \rationals^{\widehat{m} \times n}$ and $\vect{\widehat{b}} \in \rationals^{\widehat{m}\times 1}$.
	Let $\texttt{FMplex}(\sys{\widehat{a}}{\widehat{b}};\matr{{a}}\vect{x} \leq \vect{{b}},\matr{f})$ be a call in the recursion tree of the call $\texttt{FMplex}(\sys{\widehat{a}}{\widehat{b}})$ to \Cref{algo:base}, where $\matr{A}\in \rationals^{m\times n}$ and $\vect{b} \in \rationals^{m\times 1}$ (by construction $m\leq \widehat{m}$).

	Then there exist a set $\mathcal{N} \subseteq [\widehat{m}]$ and an injective mapping $\mathcal{B}: [m] \to [\widehat{m}]$ such that
	\begin{enumerate}
		\item $\sys{a}{b}$ is satisfiable if and only if $(\sys{\widehat{a}}{\widehat{b}}) \wedge (\submatr{\widehat{a}}{\mathcal{N}}\vect{x} = \subvect{\widehat{b}}{\mathcal{N}})$ is satisfiable,
		\item For each $i \in [m]$ holds $\{ \mathcal{B}(i) \} = \{ i' \in [\hat{m}] \mid f_{i,i'} \neq 0 \} \setminus \mathcal{N}$.
	\end{enumerate}
\end{thm}
\begin{proof}
	By assumption, there exists a sequence of calls of the form $\texttt{FMplex}(\sys{\widehat{a}}{\widehat{b}};\_,\_)$ starting from $\texttt{FMplex}(\sys{\widehat{a}}{\widehat{b}};\sys{\widehat{a}}{\widehat{b}},\matr{e})$ and ending in $\texttt{FMplex}(\sys{\widehat{a}}{\widehat{b}};\matr{{a}}\vect{x} \leq \vect{{b}},\matr{f})$ such that each of these calls is called by its predecessor in the sequence.
	We prove the property by induction over this sequence.

	In the base case $\texttt{FMplex}(\sys{\widehat{a}}{\widehat{b}};\sys{\widehat{a}}{\widehat{b}},\matr{e})$, we choose $\mathcal{N} = \emptyset$ and $\mathcal{B}$ as the identity.
  Then both properties are trivially fulfilled.

	Now assume $\texttt{FMplex}(\sys{\widehat{a}}{\widehat{b}};\matr{{a'}}\vect{x} \leq \vect{{b'}},\matr{f'})$ for which an $\mathcal{N}'$ and $\mathcal{B}'$ fulfilling the desired properties exist, and let $\texttt{FMplex}(\sys{\widehat{a}}{\widehat{b}};\matr{{a}}\vect{x} \leq \vect{{b}},\matr{f})$ be the succeeding call in the sequence.
	If $\matr{{a}}\vect{x} \leq \vect{{b}} = \fmpproj{\matr{a'}\vect{x} \leq \vect{b'}}{j}{\bot}$, then we simply choose $\mathcal{N} = \mathcal{N}'$ and $\mathcal{B}=\mathcal{B}'$.
	The first property holds by \Cref{lma:restricted-projection}, the second property is trivial.

	Otherwise, $(\matr{{a}}\vect{x} \leq \vect{{b}}) = \fmpproj{\matr{a'}\vect{x} \leq \vect{b'}}{j}{i}$ for some row $i$.
	We choose $\mathcal{N} := \mathcal{N}' \cup \{ \mathcal{B}'(i) \}$ and for every row $i'$ of $\matr{{a}}\vect{x} \leq \vect{{b}}$, we choose $\mathcal{B}(i'):=\mathcal{B}'(k)$ for the unique $k$ such that the row $i'$ of $\matr{{a}}\vect{x} \leq \vect{{b}}$ is the linear combination of the rows $i$ and $k$ of $\matr{{a'}}\vect{x} \leq \vect{{b'}}$ (according to the restricted projection).
	The second property obviously holds.
	We prove the first property:
	\begin{align*}
		\matr{{a}}\vect{x} \leq \vect{{b}} \text{ satisfiable} & \iff (\matr{{a'}}\vect{x} \leq \vect{{b'}}) \wedge (\rvect{a'}{i}\;\vect{x} = b'_i) \text{ satisfiable} \\
		& \iff (\sys{\widehat{a}}{\widehat{b}}) \wedge (\submatr{\widehat{a}}{\mathcal{N}'} \vect{x} = \subvect{\widehat{b}}{\mathcal{N}'}) \wedge (\rvect{a'}{i}\;\vect{x} = b'_i)   \text{ satisfiable} \\
		& \iff (\sys{\widehat{a}}{\widehat{b}}) \wedge (\submatr{\widehat{a}}{\mathcal{N}'} \vect{x} = \subvect{\widehat{b}}{\mathcal{N}'}) \wedge (\vect{\widehat{a}}_{\mathcal{B}'(i),-}\vect{x} = \widehat{b}_{\mathcal{B}'(i)})  \text{ satisfiable} \\
		& \iff (\sys{\widehat{a}}{\widehat{b}}) \wedge (\submatr{\widehat{a}}{\mathcal{N}} \vect{x} = \subvect{\widehat{b}}{\mathcal{N}})  \text{ satisfiable}
	\end{align*}

	\noindent The first equivalence holds due to \Cref{lma:restricted-projection}, the second due to the induction hypothesis.
\end{proof}

We call the above defined set $\mathcal{N}$ the \emph{non-basis}, inspired from the analogies to the simplex algorithm (discussed in \Cref{sec:simplex}).
By the above theorem, the order in which a non-basis is constructed has no influence on the satisfiability of the induced sub-problem.
In particular:

\begin{thm} \label{thm:redundancies}
	Let $\matr{a} \in \rationals^{m \times n}$, $\vect{b} \in \rationals^{m\times 1}$, $j \in [n]$, and let $i,i' \in [m]$ with $a_{i,j} \neq 0$ and $a_{i',j} \neq 0$.
	If $\fmpproj{\matr{a}\vect{x} \leq \vect{b}}{j}{i}$ is unsatisfiable, then $\fmpproj{\matr{a}\vect{x} \leq \vect{b}}{j}{i'} \wedge  (\rvect{a}{i}\;\vect{x} = b_i) $ is unsatisfiable.
\end{thm}
\begin{proof}
	By \Cref{thm:basis-nonbasis-correspondence}, if $\fmpproj{\matr{a}\vect{x} \leq \vect{b}}{j}{i}$ is unsatisfiable, then $(\matr{a}\vect{x} \leq \vect{b}) \wedge  (\rvect{a}{i}\vect{x} = b_i)$ is unsatisfiable, and trivially $(\matr{a}\vect{x} \leq \vect{b}) \wedge  (\rvect{a}{i}\vect{x} = b_i) \wedge (\rvect{a}{i'}\vect{x} = b_{i'})$ is unsatisfiable as well.
	Using \Cref{thm:basis-nonbasis-correspondence} in the other direction yields that $(\matr{a}\vect{x} \leq \vect{b}) \wedge (\rvect{a}{i'}\vect{x} = b_{i'})$ is equivalent to $\fmpproj{\matr{a}\vect{x} \leq \vect{b}}{j}{i'}$; thus $\fmpproj{\matr{a}\vect{x} \leq \vect{b}}{j}{i'} \wedge (\rvect{a}{i}\vect{x} = b_i )$ is unsatisfiable.
\end{proof}

Accordingly, if we know that a non-basis $\mathcal{N}$ induces an unsatisfiable sub-problem, then we also know that all supersets of $\mathcal{N}$ will do so.
We apply this insight to our algorithm:
Assume a call $\texttt{FMplex}(\sys{\widehat{a}}{\widehat{b}};\sys{a}{b},\matr{f})$, with corresponding $\mathcal{N}$ and $\mathcal{B}$, has a child call for row $i$ which returns $\unsat$ or $\partialunsat$.
Then every call in the recursion tree of $\texttt{FMplex}(\sys{\widehat{a}}{\widehat{b}};\sys{a}{b},\matr{f})$ whose corresponding non-basis contains $\mathcal{B}(i)$ will also return $\unsat$ or $\partialunsat$.
Hence, after we checked the child call for row $i$, we can ignore $\mathcal{B}(i)$ as designated bound in the remaining recursion tree of $\texttt{FMplex}(\sys{\widehat{a}}{\widehat{b}};\sys{a}{b},\matr{f})$.

\begin{exa}\label{example:nonbasis}
	Consider the system from \Cref{example:fmplex-idea}, extended by $c_5 : (-x_2 \leq 0)$, depicted in the top left of the \Cref{fig:nonbasis} below. That figure also shows a (partial) trace of our algorithm.
	In front of each row is indicated how it was constructed from the original constraints; $\mathcal{B}$ would give the row corresponding to this combination's first term.
	Note that we omit the eliminated variables and corresponding zero-columns.

	\begin{figure}[h]
		\begin{tikzpicture}
			\node[draw, align=center] (root) at (0,0) {
				Input\\
				$
				\begin{blockarray}{ccc}
					\begin{block}{c[cc]}
						c_1 &  -1 & -1 \bigstrut[t] \\
						c_2 &   0 &  -2 \\
						c_3 &  -2 &  1 \\
						c_4 &   0 &  1 \\
						c_5 &   0 &  -1 \bigstrut[b]\\
					\end{block}
				\end{blockarray}
				\cdot
				\left[ {\begin{array}{cc}
					x_1 \\ x_2
				\end{array} } \right] 
				\leq
				\begin{blockarray}{c}
					\begin{block}{[c]}
					-4 \bigstrut[t] \\
					-2 \\
					1 \\
					5 \\
					0 \bigstrut[b]\\
					\end{block}
				\end{blockarray}
			$\vspace{-1em}};

			\node[draw, align=left] (A1) at (0,-3.5) {$ $\\ \vspace{-1.2em}
			$
				\begin{blockarray}{rc}
					\begin{block}{r[c]}
						c_1 - c_5 &  -1 \bigstrut[t] \\
						\frac{1}{2}c_2 - c_5 &  0 \\
						c_3 + c_5 & -2 \\
						c_4 + c_5 &  0 \bigstrut[b]\\
					\end{block}
				\end{blockarray}
				\cdot
				\left[ {\begin{array}{c}
					x_1 
				\end{array} } \right] 
				\leq
				\begin{blockarray}{c}
					\begin{block}{[c]}
						-4 \bigstrut[t] \\
						-1 \\
						1 \\
						5 \bigstrut[b]\\
					\end{block}
				\end{blockarray}
			$};

			\node[draw, align=left] (A21) at (8.5,-3.5) {$ $\\ \vspace{-1.2em}
			$
			\begingroup
			\renewcommand*{\arraystretch}{1.25}
				\begin{blockarray}{rc}
					\begin{block}{r[c]}
						\frac{1}{2}c_2 \phantom{~+ 2c_1~} - c_5 &  0 \bigstrut[t] \\
						\frac{1}{3}c_3 - \frac{2}{3}c_1 + c_5&  0\\
						c_4 \phantom{~+ 2c_1} + c_5 & 0 \bigstrut[b]\\
					\end{block}
				\end{blockarray}
				\cdot
				\left[ {\begin{array}{c}
					x_1
				\end{array} } \right] 
				\leq
				\begin{blockarray}{c}
					\begin{block}{[c]}
						-1 \bigstrut[t] \\
						3 \\
						5 \bigstrut[b]\\
					\end{block}
				\end{blockarray}
			\endgroup
			$};

			\node[draw, align=left] (A2) at (8.5,0) {$ $\\ \vspace{-1.2em}
			$
				\begin{blockarray}{rc}
					\begin{block}{r[c]}
						\frac{1}{2}c_2 - c_1 &  1 \bigstrut[t] \\
						c_3 + c_1 &  -3\\
						c_4 + c_1 & -1 \\
						c_5 - c_1 & 1 \bigstrut[b]\\
					\end{block}
				\end{blockarray}
				\cdot
				\left[ {\begin{array}{c}
					x_1
				\end{array} } \right] 
				\leq
				\begin{blockarray}{c}
					\begin{block}{[c]}
						3 \bigstrut[t] \\
						-3 \\
						1 \\
						4 \bigstrut[b]\\
					\end{block}
				\end{blockarray}
			$};

			\draw[{Latex[length=2mm,width=2mm]}-] (A1) -- (root) node[pos=0.5,left=0.3cm, align=center] {elim. $x_2$ by $c_5$} node[pos=0.5,right=0.3cm] {$\mathcal{N} = \{5\}$};

			\draw[{Latex[length=2mm,width=2mm]}-] (A2) -- (root) node[pos=0.5,above=0.1cm, align=center] {elim. $x_2$\\by $c_1$} node[pos=0.5,below=0.3cm, align=left] {$\mathcal{N} = \{1\}$};

			\draw[{Latex[length=2mm,width=2mm]}-] (A21) -- (A2) node[pos=0.5,left=0.3cm, align=left] {elim. $x_1$ by $c_5$} node[pos=0.5,right=0.3cm, align=left] {$\mathcal{N} = \{1,5\}$};

		\end{tikzpicture}
		\caption{Algorithm trace described in \Cref{example:nonbasis}.}\label{fig:nonbasis}
	\end{figure}
	If $c_5$ is tried first as greatest lower bound on $x_2$, then the resulting partial projection contains a local conflict.
	Thus, this branch and, due to \Cref{thm:redundancies}, any non-basis containing row $5$ yields an unsatisfiable system.

	Next, we try $c_1$ as greatest lower bound on $x_2$, giving the branch on the right.
	If we now choose $(x_1 \leq 4)$, which corresponds to $c_5$ via $\mathcal{N}$ and $\mathcal{B}$, as smallest upper bound on $x_1$, another local conflict occurs (bottom right, row 1).
	As $5$ is contained in the non-basis of that system, we could have predicted this and thus avoid computing this branch.
\end{exa}

We update the $\texttt{FMplex}$ algorithm as shown in \refalgobounds using the following definition:

\begin{defi}
	The set of \emph{branch choices} for $\sys{a}{b}$ with $m$ rows w.r.t. $I \subseteq [m]$ is
	\begin{align*}
		\textit{branch\_choices}(\sys{a}{b}, I) := \{(x_j, V\setminus I) \mid (x_j, V) \in \textit{branch\_choices}(\sys{a}{b})\}.
	\end{align*}
\end{defi}

\noindent This modification clearly prevents visiting the same non-basis twice in the following sense:

\begin{thm}
	Let $\texttt{FMplex}(\sys{\widehat{a}}{\widehat{b}};\matr{a}\vect{x} \leq \vect{b},\_,\mathcal{N},\_)$ and $\texttt{FMplex}(\sys{\widehat{a}}{\widehat{b}};\matr{a'}\vect{x} \leq \vect{b'},\_,\mathcal{N}',\_)$ be two calls in the recursion tree of a call to \refalgobounds\!\!.
	Then either $\mathcal{N} \neq \mathcal{N'}$ or one of the systems occurs in the subtree below the other and only unbounded variables are eliminated between them (i.e. one results from the other by deleting some rows).\qed
\end{thm}

\Cref{thm:unique-base-termination} states that, still, every initial call to \refalgobounds terminates with $\sat$ or a global conflict.
This follows by a slight modification of the proof of \Cref{thm:global-conflict}.

\begin{thm}\label{thm:unique-base-termination}
	Let $\matr{\widehat{a}} \in \rationals^{\widehat{m} \times n}$, and $\vect{\widehat{b}} \in \rationals^{\widehat{m}\times 1}$.
	Then $\sys{\widehat{a}}{\widehat{b}}$ is unsatisfiable if and only if the call $\texttt{FMplex}(\sys{\widehat{a}}{\widehat{b}})$ to \refalgobounds terminates with a global conflict.
\end{thm}
\begin{proof}
	Consider the proof of \Cref{thm:global-conflict} again: We showed that there exists a path in the recursion tree leading to a global conflict.
	We now show that we can still construct such a path.
	We first fix the earliest path in the recursion tree that leads \Cref{algo:base} to a global conflict, which does exist according to \Cref{thm:global-conflict}; let $\widehat{K} \subseteq [\widehat{m}]$ be the corresponding subset inducing the minimal unsatisfiable subset.
	Along this path, all but one element of $\widehat{K}$ are chosen into the non-basis, as can be deduced from the proof of \Cref{thm:global-conflict}.
	Without loss of generality, we can assume that this path does not insert anything else into the non-basis (otherwise, it would not be the earliest global conflict, or there would be variables not occurring in the conflict, whose elimination is irrelevant).

	We now execute \refalgobounds on the same input.
	Towards a contradiction, we assume that there is a state on the given path where the choice to the next element is forbidden in \refalgobounds\!\!, that is, an element $k \in \widehat{K}$ is ignored ($k \in I$).
	This is only possible if there has been an earlier call in the execution trace of \refalgobounds where $k$ was added to the non-basis, this call returned $\partialunsat$, and its parent call is on the given path (otherwise, $k$ would already have been removed from the set of ignored rows). 
	But then, still only elements from $\widehat{K}$ were in the non-basis of that call.
	But this means, there would be a earlier path that leads to the same global conflict $\widehat{K}$, which is a contradiction to the above assumption.
\end{proof}

\begin{algorithm}[p]%
	\linespread{1.2}\selectfont
	\caption{$\texttt{FMplex}(\sys{\widehat{a}}{\widehat{b}};\sys{a}{b},\matr{f},$\chng{$\mathcal{N},\mathcal{B},I$}$,$\chngd{$\texttt{lvl},\texttt{bt\_lvl}$}$)$}\label{algo:combined}

	  \begin{description}
		\item[Algorithm 2a] \label{algo:combined-base} The base method consists of the plain parts (as in \Cref{algo:base}).
		\item[Algorithm 2b] \label{algo:leaving-out-bounds} Consists of the base method and the \chng{framed parts}.
		\item[Algorithm 2c] \label{algo:backtracking} Consists of the base method, the \chng{framed} and the \chngd{filled} parts.
	\end{description}
	\hrulefill\\
	\SetKwInOut{Data}{Data}
	\SetKwInOut{Input}{Input}
	\SetKwInOut{Output}{Output}
	\DontPrintSemicolon
	\SetKwFunction{FCheck}{FMplex}
	\SetKwProg{Fn}{Function}{:}{}
	\Data{$\matr{\widehat{a}} \in \rationals^{\widehat{m} \times n},\ \vect{\widehat{b}}\in \rationals^{\widehat{m}}$}
	\Input{$\matr{a} \in \rationals^{m \times n},\ \vect{b}\in \rationals^{m}$, $\matr{f} \in \rationals^{m \times \widehat{m}}$ s.t. $\matr{f}\matr{\widehat{a}} = \matr{a}$ and $\matr{f}\vect{\widehat{b}} = \vect{b}$, \chng{$\mathcal{N}\subseteq [\widehat{m}]$, $\mathcal{B}\colon [m] \rightarrow [\widehat{m}]$, $I\subseteq [\widehat{m}]$}, \chngd{$\texttt{lvl} \in [n] \cup \{ 0 \}$, $\texttt{bt\_lvl}\colon [m] \to [n] \cup \{ 0 \}$}}
	\Output{(\sat, $\alpha$) with $\alpha \models \matr{a}\vect{x} \leq \vect{b}$, or (\unsat, $K$) where $K \subseteq [\widehat{m}]$, or\\
	(\partialunsat, \chngd{$l, K$}) \chngd{where $l \in [n]$ and $K \subseteq [\widehat{m}]$}\\\hspace*{-4em}\hrulefill}
	\lIf{$\matr{a} = 0 \land \vect{b} \geq 0$} {
		\Return{$($\sat, $())$}\tcp*[f]{trivial sat}
	}
	\If(\tcp*[f]{global conflict}){$\exists i \in [m].\; \rvect{a}{i} = 0 \land b_i < 0 \land \rvect{f}{i} \geq 0$}{
			\Return{$(\unsat, \{ i' \mid f_{i,i'} \neq 0 \})$} 
	} 
	\boxchngd{a}{0.92}
	\boxchngdstart{a}
	\If(\tcp*[f]{local conflict}){$\exists i \in [m].\; \rvect{a}{i} = 0 \land b_i < 0 \land \rvect{f}{i} \ngeq 0$}{
		$i := \arg\min_{i \in [m]}\{ \texttt{bt\_lvl}(i) \mid \rvect{a}{i} = 0 \land b_i < 0 \}$\;
		\Return $(\partialunsat,\texttt{bt\_lvl}(i)-1,\{ i' \mid f_{i,i'} \neq 0 \})$ \;
	}
	\boxchngdend{a}%
	\chngd{$K := \emptyset$} \;
	\textbf{choose} $(x_j,V) \in \textit{branch\_choices}(\matr{a}\vect{x} \leq \vect{b},\chng{$\{ \mathcal{B}^{-1}(i) \mid i \in I \}$})$\;
	\ForEach{$i \in V$}{
		\textbf{compute} $\matr{a'}\vect{x} \leq \vect{b'} := \fmpproj{\matr{a}\vect{x} \leq \vect{b}}{j}{i}$ with projection matrix $\matr{f'}$ \chngd{and backtrack levels $\texttt{bt\_lvl}'$} \;
		\chng{\textbf{compute} $\mathcal{N}'$ := ($\mathcal{N} \cup \{ \mathcal{B}(i) \}$ \textbf{ if }  $i \neq \bot$ \textbf{ else }   $\mathcal{N}$) and $\mathcal{B}'$ from $\mathcal{N}'$ and $F'$}\;
		\Switch{\FCheck$(\sys{\widehat{a}}{\widehat{b}};\matr{a'}\vect{x} \leq \vect{b'},\matr{f'}\matr{f},$\chng{$\mathcal{N}', \mathcal{B}', I$}$,\chngd{$\texttt{\upshape lvl}+1,\texttt{\upshape bt\_lvl}'$})$}{
			\lCase{$(\unsat, K')$}{\Return{$(\unsat, K')$}}
			\lCase{$(\sat, \alpha)$}{\Return{$(\sat, \alpha[x_j \mapsto r])$} for a suitable $r \in \rationals$}
			\boxchngd{b}{0.87}
			\boxchngdstart{b}
			\Case{$(\partialunsat, l, K')$}{
				\lIf{$l<\texttt{\upshape lvl}$}{
					\Return{$(\partialunsat, l, K')$}\tcp*[f]{backtrack} %
				}
				\lElse {
					$K := K \cup K'$ %
				}
			}
		}
		\boxchngdend{b}	\chng{$I$ := $I \cup \{ \mathcal{B}(i) \}$} \;\label{algo:backtracking:ignore}
	}
	\chngd{
	\label{algo:combined:lvl0} \lIf{$\texttt{\upshape lvl}=0$}{
		\Return{$(\unsat, K)$}
	}}\;
	\Return{$(\partialunsat, \chngd{\texttt{\upshape lvl-1}, K})$}

  \end{algorithm}

\subsection{Backtracking of Local Conflicts}

So far, we ignored local conflicts that witness the unsatisfiability of a given sub-problem.
In this section, we will cut off parts of the search tree based on local conflicts and examine the theoretical implications.

We applied Farkas' Lemma on conflicting rows in some sub-problem that are positive linear combinations of rows from the input system.
We can also apply Farkas' Lemma to conflicting rows which are positive linear combinations of some \emph{intermediate} system to conclude the unsatisfiability of the latter.
Whenever such a conflict occurs, we can backtrack to the parent system of that unsatisfiable system.
Instead of tracking the linear combinations of every row in terms of the rows of each preceding intermediate system, we can do an incomplete check: If a conflicting row was computed only by addition operations, then it is a positive linear combination of the involved rows.
Thus, we assign to every intermediate system a level, representing its depth in the search tree and store for every row the level where the last subtraction was applied to the row (i.e. a lower (upper) bound was subtracted from another lower (upper) bound).
If a row is conflicting, we can conclude that the intermediate system at this level is unsatisfiable, thus we can jump back to its parent.

Assume the current system is $\matr{a}\vect{x} \leq \vect{b}$ at level $\texttt{lvl}$ with $m$ rows whose backtracking levels are stored in $\texttt{bt\_lvl}: [m] \to ([n] \cup \{ 0 \})$.
If $\texttt{lvl}=0$, then $\texttt{bt\_lvl}$ maps all values to $0$.
When computing $\fmpproj{\matr{a}\vect{x} \leq \vect{b}}{j}{i}$ for some $x_j$ and $i$ with projection matrix $\matr{f}$, the backtracking levels of the rows in the resulting system $\matr{FA}\vect{x} \leq \matr{F}\vect{b}$ are stored in $\texttt{bt\_lvl}'$ where for each row $i''$
\[\texttt{bt\_lvl}'(i'') := \left.\begin{cases}
	\max\{\texttt{bt\_lvl}(i),\texttt{bt\_lvl}(i')\} & \text{ if } f_{i'',i}, f_{i'',i'} > 0 \text { and } f_{i'',k}=0,\  k \notin \{ i,i' \} \\
	\texttt{lvl} + 1 &  \text{ otherwise}.
\end{cases}\right.\]

The backtracking scheme is given in \refalgobacktracking\!, which returns additional information in the $\partialunsat$ case, that is the backtrack level $l$ of the given conflict, and a (possibly non-minimal) unsatisfiable subset $K$.

\begin{thm}
	Let $\texttt{FMplex}(\_;\matr{a}\vect{x} \leq \vect{b},\_,\_,\_,\texttt{\upshape lvl},\_)$ be a call to \refalgobacktracking\!, and consider a second call $\texttt{FMplex}(\_;\matr{a'}\vect{x} \leq \vect{b'},\_,\_,\_,\_,\texttt{\upshape bt\_lvl}')$ in the recursion tree of the first call.
	If $\matr{a'}\vect{x} \leq \vect{b'}$ has a local conflict in a row $i$ with $\texttt{\upshape bt\_lvl}'(i)=\texttt{\upshape lvl}$, then $\matr{a}\vect{x} \leq \vect{b}$ is unsatisfiable.
\end{thm}
\begin{proof}
	By construction of $\texttt{bt\_lvl}$', $\rvect{a'}{i}\;\vect{x} \leq b'_i$ is a positive sum of rows from $\matr{a}\vect{x} \leq \vect{b}$, i.e. there exists an $\vect{f} \in \rationals^{1\times m}$ such that $(\vect{f} \matr{a} \vect{x} \leq  \vect{f} \vect{b}) = (\rvect{a'}{i}\;\vect{x} \leq b'_i)$.
	Then by Farkas' Lemma, $\matr{a}\vect{x} \leq \vect{b}$ is unsatisfiable.
\end{proof}

While it is complete and correct, \refalgobacktracking does not always terminate with a \emph{global} conflict (i.e. \Cref{thm:global-conflict} does not hold anymore), explaining the need for \Cref{algo:combined:lvl0}:

\begin{exa}\label{example:backtracking}
	We use \refalgobacktracking to check the satisfiability of the system depicted in the top left of \Cref{fig:backtracking} below.
	That figure also shows a trace of our algorithm, and to the left of each row is indicated how it was constructed from the original constraints.

	\begin{figure}[h]
		\begin{tikzpicture}
			\node[draw, align=center] (root) at (0,0) {
				Input, $\texttt{lvl} = 0$\\
				$
				\begin{blockarray}{cccc}
					\begin{block}{c[ccc]}
						c_1 &  1 & -1 & -1 \bigstrut[t] \\
						c_2 &  0 &  0 & -1 \\
						c_3 &  0 & -1 &  1 \\
						c_4 & -1 &  1 &  0 \\
						c_5 &  1 &  0 &  0 \bigstrut[b]\\
					\end{block}
				\end{blockarray}
				\cdot
				\left[ {\begin{array}{ccc}
					x_1 \\ x_2 \\ x_3
				\end{array} } \right] 
				\leq
				\begin{blockarray}{c}
					\begin{block}{[c]}
					0 \bigstrut[t] \\
					0 \\
					0 \\
					-1 \\
					-1 \bigstrut[b]\\
					\end{block}
				\end{blockarray}
			$\vspace{-1em}};

			\node[draw, align=left] (A1) at (0,-3.5) {$ $\\ \vspace{-1.2em}
			$
				\begin{blockarray}{lcc}
					\begin{block}{l[cc]}
						c_2 - c_1 &  -1 & 1 \bigstrut[t] \\
						c_3 + c_1 &  1 & -2 \\
						c_4 & -1 &  1 \\
						c_5 &  1 & 0 \bigstrut[b]\\
					\end{block}
				\end{blockarray}
				\cdot
				\left[ {\begin{array}{ccc}
					x_1 \\ x_2
				\end{array} } \right] 
				\leq
				\begin{blockarray}{c}
					\begin{block}{[c]}
						0 \bigstrut[t] \\
						0 \\
						-1 \\
						-1 \bigstrut[b]\\
					\end{block}
				\end{blockarray}
			$};

			\node[draw, align=left] (A11) at (0,-6.25) {$ $\\ \vspace{-1.2em}
			$
				\begin{blockarray}{lc}
					\begin{block}{l[c]}
						2c_2 - c_1 + c_3 & - 1 \bigstrut[t] \\
						2c_4 + c_1 + c_3 & - 1 \\
						\phantom{2}c_5 &  1 \bigstrut[b]\\
					\end{block}
				\end{blockarray}
				\cdot
				\left[ {\begin{array}{ccc}
					x_1
				\end{array} } \right] 
				\leq
				\begin{blockarray}{c}
					\begin{block}{[c]}
						0 \bigstrut[t] \\
						-2 \\
						-1 \bigstrut[b]\\
					\end{block}
				\end{blockarray}
			$};

			\node[draw, align=left] (A111) at (8.75,-6.25) {$ $\\ \vspace{-1.2em}
			$
				\begin{blockarray}{lc}
					\begin{block}{l[c]}
						2c_4 \phantom{+c_3} + 2c_1 - 2c_2& 0 \bigstrut[t] \\
						\phantom{2}c_5 + c_3 - c_1 + 2c_2 & 0 \bigstrut[b]\\
					\end{block}
				\end{blockarray}
				\leq
				\begin{blockarray}{c}
					\begin{block}{[c]}
						-2 \bigstrut[t]\\
						-1 \bigstrut[b]\\
					\end{block}
				\end{blockarray}
			$};

			\node[draw, align=left] (A2) at (8.5,0) {$ $\\ \vspace{-1.2em}
			$
				\begin{blockarray}{lcc}
					\begin{block}{l[cc]}
						c_1 - c_2 &  1 & -1 \bigstrut[t] \\
						c_3 + c_2 &  0 & -1 \\
						c_4 & -1 &  1 \\
						c_5 &  1 &  0 \bigstrut[b]\\
					\end{block}
				\end{blockarray}
				\cdot
				\left[ {\begin{array}{ccc}
					x_1 \\ x_2
				\end{array} } \right] 
				\leq
				\begin{blockarray}{c}
					\begin{block}{[c]}
						0 \bigstrut[t] \\
						0 \\
						-1 \\
						-1 \bigstrut[b]\\
					\end{block}
				\end{blockarray}
			$};

			\node[draw, align=left] (A21) at (8.5,-3) {$ $\\ \vspace{-1.2em}
			$
				\begin{blockarray}{rc}
					\begin{block}{r[c]}
						c_1 - c_2 + c_4 & 0 \bigstrut[t] \\
						c_3 + c_2 \phantom{~ + c_4} & -1 \\
						c_5 + \phantom{c_2 + ~} c_4 &  1 \bigstrut[b]\\
					\end{block}
				\end{blockarray}
				\cdot
				\left[ {\begin{array}{ccc}
					x_1
				\end{array} } \right] 
				\leq
				\begin{blockarray}{c}
					\begin{block}{[c]}
						-1 \bigstrut[t] \\
						0 \\
						-2 \bigstrut[b]\\
					\end{block}
				\end{blockarray}
			$};

			\draw[{Latex[length=2mm,width=2mm]}-] (A1) -- (root) node[pos=0.5,left=0.3cm, align=center] {elim. $x_3$ by $c_1$} node[pos=0.5,right=0.3cm] {$\mathcal{N} = \{1\}$};

			\draw[{Latex[length=2mm,width=2mm]}-] (A2) -- (root) node[pos=0.5,above=0.1cm, align=center] {elim. $x_3$\\by $c_2$} node[pos=0.5,below=0.3cm, align=left] {$\mathcal{N} = \{2\}$};

			\draw[{Latex[length=2mm,width=2mm]}-] (A11) -- (A1) node[pos=0.5,left=0.3cm, align=left] {elim. $x_2$ by $c_3$} node[pos=0.5,right=0.3cm, align=left] {$\mathcal{N} = \{1,3\}$};

			\draw[{Latex[length=2mm,width=2mm]}-] (A21) -- (A2) node[pos=0.5,left=0.3cm, align=left] {elim. $x_1$ by $c_4$} node[pos=0.5,right=0.3cm, align=left] {$\mathcal{N} = \{2,4\}$};

			\draw[{Latex[length=2mm,width=2mm]}-] (A111) -- (A11) node[pos=0.5,above=0.1cm, align=center] {elim. $x_1$\\by $c_2$} node[pos=0.5,below=0.3cm, align=left] {$\mathcal{N} = \{1,3,2\}$};
		\end{tikzpicture}
		\caption{Algorithm trace described in \Cref{example:backtracking}.}\label{fig:backtracking}
	\end{figure}

	We start by eliminating $x_3$, using the lower bounds given by rows 1 and 2.
	First, $c_1$ is tried as greatest lower bound, leading to the branch on the left.
	Following that path leads to the system in the bottom right, which contains two local conflicts (local, because the linear combinations contain negative coefficients).
	The first row has backtrack-level 3, as it resulted from combining two lower bounds on level 2.
	The second row has backtrack-level 1, as it resulted from the system on level 1 by only combining lower with upper bounds.
	This means that the ancestor call at level $1$ is unsatisfiable, and thus we jump back to level $0$.

	The second branch for $x_3$ is visited (on the right), and we now first eliminate $x_1$ instead of $x_2$.
	There is only one lower bound, the row corresponding to $c_4$, and the resulting child system contains a local conflict, so this branch returns $\partialunsat$.
	Now, both branches of the initial call returned $\partialunsat$, but no global conflict was found.

	Note that the optimization from \Cref{subsec:avoiding-redundancies} is irrelevant in this example; backjumping alone suffices to break \Cref{thm:global-conflict}.
\end{exa}

\subsection{Strict Constraints}\label{sec:strict}
So far, we only considered conjunctions of weak linear constraints as input.
In practice, however, we often want to admit strict constraints. In the context of SMT solving, strict constraints cannot be avoided, as they appear as the negation of weak constraints which need to be satisfied due to the Boolean structure.
For example, an input $\sys{a}{b} \land \lnot (\vect{c}\vect{x} \leq d)$ is equivalent to the system $\sys{a}{b} \land (-\vect{c}\vect{x} < -d)$.

For the Fourier-Motzkin method, handling strict constraints is straightforward:
Upper and lower bounds are combined as before, with the addition that the combination of two constraints $c_1$ and $c_2$ is a strict constraint if and only if at least one of $c_1$, $c_2$ is strict.
The intuition here is that from $(l < x) \land (x \leq u)$ follows $(l < u)$, but not necessarily $l \leq u$.

Transferring this idea to \texttt{FMplex} is only straightforward for combinations of lower and upper bounds, but we cannot simply apply the same reasoning for combinations of two lower or two upper bounds.
The following example shows how this would lead to incorrect results.

\begin{exa}\label{example:strict}
    Consider the algorithm trace given below in \Cref{fig:strict}.
    The first constraint ($c_1$) is strict.
    If we therefore assume that the combination of $c_1$ and $c_2$ must yield a strict constraint, then we end up with a trivially false constraint $(0<0)$ after eliminating both variables $x_1$ and $x_2$.
    This constraint is a positive linear combination of the input and thus the algorithm would report \unsat.
    This is incorrect, as $x_1 = x_2 = 1$ is indeed a model.
    
    Note that the alleged global conflict $(0<0)$ is actually a linear combination of only the weak constraints $c_2$ and $c_3$, and the coefficients for the strict constraint $c_1$ cancel out in the elimination of $x_2$.
    But $(0 < 0)$ is not a consequence of $c_2$ and $c_3$ and the linear combination actually equates to $(0 \leq 0)$.

    \begin{figure}[h]
		\begin{tikzpicture}
			\node[draw, align=center] (root) at (0,0) {
				Input\\
				$
				\begin{blockarray}{ccc}
					\begin{block}{c[cc]}
						c_1 &  -1 &  0 \bigstrut[t] \\
						c_2 &  -1 &  1 \\
						c_3 &   1 & -1 \bigstrut[b]\\
					\end{block}
				\end{blockarray}
				\cdot
				\left[ {\begin{array}{cc}
					x_1 \\ x_2
				\end{array} } \right] 
				\begin{blockarray}{cc}
					\begin{block}{c[c]}
                        <    & 0 \bigstrut[t] \\
                        \leq & 0 \\
                        \leq & 0 \bigstrut[b]\\
					\end{block}
				\end{blockarray}
			$\vspace{-1em}};

			\node[draw, align=left] (A1) at (6.3,-0.4) {$ $\\ \vspace{-1.2em}
			$
				\begin{blockarray}{rr}
					\begin{block}{r[r]}
						c_2 - c_1 &  1 \cdot x_2 \bigstrut[t] \\
						c_3 + c_1 & -1 \cdot x_2 \bigstrut[b]\\
					\end{block}
				\end{blockarray}
				\begin{blockarray}{cc}
					\begin{block}{c[c]}
						< & 0 \bigstrut[t] \\
						< & 0 \bigstrut[b]\\
					\end{block}
				\end{blockarray}
			$};

            \node[draw, align=left] (A2) at (11.3,-0.4) {$c_3 + c_2 [0 < 0]$};

			\draw[{Latex[length=2mm,width=2mm]}-] (A1) -- (root) node[pos=0.35,above=0.1cm, align=left] {elim. $x_1$\\by $c_1$};

			\draw[{Latex[length=2mm,width=2mm]}-] (A2) -- (A1) node[pos=0.35,above, align=left] {elim. $x_2$\\by $c_2$};
		\end{tikzpicture}
		\caption{Algorithm trace described in \Cref{example:strict}.}\label{fig:strict}
	\end{figure}
\end{exa}

Instead, we use a well-known way of accomodating strict constraints described e.g. in described in \cite{nalbach2021extending} and \cite{king_effective_2014}, which is to transform them into weak constraints by adding a new variable (which we call $\delta$ here) representing an infinitesimal positive value.

\begin{lem}[Elimination of Strict Constraints]
    Let $\matr{a} \in \rationals^{m \times n}$, $\vect{b} \in \rationals^{m}$, $\matr{a}' \in \rationals^{k\times n}$ and $\vect{b}' \in \rationals^k$ for some $m,k,n \in \nats$.
    Then
    \[
        \left(\sys{a}{b} \land \bigwedge_{i=1}^k (\rvect{a'}{i}\vect{x} < b'_i)\right)\quad \equiv\quad \exists \delta. \left(\sys{a}{b} \land \bigwedge_{i=1}^k (\rvect{a'}{i}\vect{x} + \delta \leq b'_i) \land (\delta > 0)\right). \tag*{\qed}
    \]
\end{lem}

The symbol $\delta$ is treated as a variable, but we do not eliminate it in the \texttt{FMplex} algorithm.
Hence, the single remaining strict constraint $\delta > 0$ does not pose any problems.
When \texttt{FMplex} reaches a leaf, all other variables have been eliminated, and only constant bounds on $\delta$ and trivial constraints are left, making it easy to decide whether that leaf is satisfiable.
If it is satisfiable, we can choose any value for $\delta$ within the bounds and construct the rest of the model as before.

We can deduce \unsat~ from global conflicts as before, but we can now also count constraints equivalent to $(\delta \leq 0)$ as global conflicts if they are constructed as a non-negative linear combination of the input system, because they contradict the assertion $(\delta > 0)$.

\section{Relation to Other Methods}
\label{sec:othermethods}

\subsection{Simplex Algorithm}
\label{sec:simplex}

The simplex method \cite{dantzig1998linear,lemke1954dual} is an algorithm for linear optimization over the reals and is able to solve \emph{linear programs}.
The \emph{general simplex} \cite{dutertre2006integrating} is an adaption for checking the satisfiability of systems of linear constraints.
We illustrate its idea for the weak case.

Remind that given a system $\sys{a}{b}$ with $m$ rows, by the fundamental theorem of linear programming (\Cref{thm:linearprogramming}), $\sys{a}{b}$ is satisfiable if and only if there exists some maximal subset $\mathcal{N} \subseteq [m]$ such that $\submatr{a}{\mathcal{N}}$ is linearly independent and $\sys{a}{b} \cup \submatr{a}{\mathcal{N}}\vect{x}=\subvect{b}{\mathcal{N}}$ is satisfiable - the latter can be checked algorithmically using Gaussian elimination, resulting in a system where each variable is replaced by bounds induced by the rows $\mathcal{N}$.
This system along with the information which element in $\mathcal{N}$ was used to eliminate which variable is called a \emph{tableau}.
The idea of the simplex method is to do a local search on the set $\mathcal{N}$ (called \emph{non-basis}).
That is, we replace some $i \in \mathcal{N}$ (\emph{leaving variable}) by some $i' \in [m] \setminus \mathcal{N}$ (\emph{entering variable}) obtaining $\mathcal{N}' := \mathcal{N} \cup \{ i' \} \setminus \{ i \}$ such that $\submatr{a}{\mathcal{N}'}$ is still linearly independent.
The clou is that the tableau representing $(\sys{a}{b}) \wedge (\subsyseq{a}{b}{\mathcal{N}})$ can be efficiently transformed into $(\sys{a}{b}) \wedge (\subsyseq{a}{b}{\mathcal{N}'})$ (called \emph{pivot operation}), and progress of the local search can be achieved by the choice of $i$ and $i'$.
These local search steps are performed until a satisfying solution is found, or a conflict is found.
These conflicts are detected using Farkas' Lemma (\Cref{thm:farkas}), i.e. a row in the tableau induces a trivially false constraint and is a positive linear combination of some input rows.

As suggested by \Cref{thm:basis-nonbasis-correspondence}, there is a strong correspondence between a tableau of the simplex algorithm and the intermediate systems constructed in FMplex.
More precisely, if a non-basis of a simplex tableau is equal to the non-basis of a leaf system of \Cref{algo:base}, then the tableau is satisfiable if and only if the FMplex system is satisfiable.
In fact, we could use the same data structure to represent the algorithmic states.
Comparing the two algorithms structurally, FMplex explores the search space in a tree-like structure using backtracking, while simplex can jump between neighbouring leaves directly.

The idea for \refalgobounds that excludes visiting the same non-basis in fact arose from the analogies between the two methods.
Further, we observe a potential advantage of FMplex: Simplex has more non-bases reachable from a given initial state than the leaves of the search tree of FMplex, as FMplex needs only to explore all lower or all upper bounds of a variable while simplex does local improvements blindly.
Heuristically, simplex cuts off large parts of its search space and we expect it often visits fewer non-bases than FMplex - however, as the pruning done by FMplex is by construction of the algorithm, we believe that there might be combinatorially hard instances on which it is more efficient than simplex.

\subsection{Virtual Substitution Method}

\emph{Virtual substitution} \cite{loosApplyingLinearQuantifier1993,weispfenning1997quantifier} admits  quantifier elimination for real arithmetic formulas.
Here, we consider its application on existentially quantified conjunctions of linear constraints.

The underlying observation is that the satisfaction of a formula changes at the zeros of its constraints and is invariant between the zeros.
Thus, the idea is to collect all \emph{symbolic zeros} $\text{zeros}(\varphi)$ of all constraints in some input formula $\varphi$.
If all these constraints are weak, then a variable $x_j$ is eliminated by plugging every zero and an arbitrarily small value $-\infty$ into the formula, i.e. $\exists x_j.\; \varphi$ is equivalent to $\varphi[-\infty\vs x_j] \vee \bigvee_{\xi \in \text{zeros}(\varphi)} \varphi[\xi\vs x_j]$.
The formula $\varphi[t\vs x_j]$ encodes the semantics of substituting the term $t$ for $x_j$ into the formula $\varphi$ (which is a disjunction of conjunctions).
As we can pull existential quantifiers into disjunctions, we can iteratively eliminate multiple variables by handling each case separately.

The resulting algorithm for quantifier elimination is singly exponential; further optimizations (\cite{nipkowLinearQuantifierElimination2008a} even proposes to consider only lower or upper bounds for the test candidates) lead to a procedure very similar to the FMplex quantifier elimination: Substituting a test candidate into the formula is equivalent to computing the restricted projection w.r.t. a variable bound.
However, our presentation allows to exploit the correspondence with the FM method.
Moreover, we handle the special case that there are only lower (resp. upper) bounds separately, while the virtual substitution handles it by always computing $\varphi[-\infty\vs x_j]$.

Virtual substitution can also be adapted for SMT solving \cite{vssmt} to a depth-first search similar to FMplex.
A conflict-driven search for virtual substitution on conjunctions of weak linear constraints has been introduced in \cite{korovin2014towards}, which tracks intermediate constraints as linear combinations of the input constraints similarly to FMplex.
Their conflict analysis is a direct generalization of the global conflicts in FMplex and is thus slightly stronger than our notion of local conflicts.
However, their method requires storing and maintaining a lemma database, while FMplex stores all the information for pruning the search tree locally.
The approaches have strong similarities, although they originate from quite different methods.
Further, our presentation shows the similarities to simplex, is easily adaptable for strict constraints, and naturally extensible to work incrementally.

\subsection{Sample-Based Methods}
There exist several depth-first search approaches, including McMillan et al. \cite{mcmillanGeneralizing2009}, Cotton \cite{cottonNatural2010} and Korovin et al. \cite{korovinConflict2009,korovinSolving2011}, which maintain and adapt a concrete (partial) variable assignment.
They share the advantage that combinations of constraints are only computed to guide the assignment away from an encountered conflict, thus avoiding many unnecessary computations, compared to FM.

Similar to FMplex, these methods perform a search with branching, backtracking and learning from conflicting choices.
However, they branch on variable assignments, with infinitely many possible choices in each step.
Interestingly, the bounds learned from encountered conflicts implicitly partition the search space into a finite number of regions to be tried, similar to what FMplex does explicitly.
In fact, we deem it possible that \cite{korovinConflict2009} or \cite{korovinSolving2011} try and exclude assignments from exactly the same regions that FMplex would visit (even in the same order).
However, the sample-based perspective offers different possibilities for heuristic improvements than FMplex: choosing the next assigned value vs. choosing the next lower bound; deriving constant variable bounds vs. structural exploits using Farkas' Lemma; potential for rapid solutions vs. control and knowledge about the possible choices.

Moreover, these methods offer no straight-forward adaption for quantifier elimination, while FMplex does.
However, \cite{mcmillanGeneralizing2009} and \cite{cottonNatural2010} can handle not only conjunctions, but any quantifier-free LRA formula in conjunctive normal form.

\section{Experimental Evaluation}
\label{sec:experiments}
We implemented several heuristic variants of the FMplex algorithm, as well as the generalized \emph{simplex} and the \emph{FM} methods as non-incremental DPLL(T) theory backends in our SMT-RAT solver \cite{corziliusSMTRAT2015} and compared their performance in the context of satisfiability checking.
Using the transformation given in \cite{nalbach2021extending} and case splitting as in \cite{barrett2006splitting}, we extended the method to also handle strict and not-equal-constraints.

The base version of FMplex, \Cref{algo:base}, was tested with two different heuristics for the choice of the eliminated variable and for the order in which the branches are checked.
These choices may strongly influence the size of the explored search tree; in the best case, the very first path leads to a satisfiable leaf or to a global conflict.

\paragraph{Min-Fanout} We greedily minimize the number of children: for any $\sys{a}{b}$ and $I$, we choose $(x_j,V) \in \textit{branch\_choices}(\sys{a}{b},I)$ such that $|V|$ is minimal; in case that this minimum is $1$, we prefer choices with $V = \{\bot\}$ over the other choices.

We prefer rows with a lower (earlier) backtrack level, motivated by finding a global conflict through trying positive linear combinations first.
Moreover, if backtracking is used then we expect this heuristic to allow for backtracking further back on average.

\paragraph{Min-Column-Length} A state-of-the-art heuristic for simplex in the context of SMT solving is the \emph{minimum column length} \cite{soisimplex}: we choose the variables for leaving and entering the non-basis such that the number of necessary row operations is minimized.
We resemble this heuristic in FMplex as follows: we prefer choices $(x_j, \{\bot\})$ and if there is no such $j$, we take the $j \in [n]$ with minimal $|\lbs{a}{j}|+|\ubs{a}{j}|$ and take the smallest choice between $\lbs{a}{j}$ and $\ubs{a}{j}$.

We first choose the rows which have the least non-zero coefficients (i.e. contain the least variables) to prefer sparse sub-problems.
This can be understood as \emph{Min-Row-Length}.

\vspace{0.5em}
We consider the following solver variants: \texttt{FMplex-a-MFO} and \texttt{FMplex-a-MCL} implement \refalgobase (i.e. \Cref{algo:base}) with the Min-Fanout and the Min-Column-Length heuristic, respectively.
\texttt{FMplex-a-Rand-1/2} denotes two variants of \refalgobase where all choices are taken pseudo-randomly with different seeds.
\texttt{FMplex-b-MFO} implements \refalgobounds and \texttt{FMplex-c-MFO} implements \refalgobacktracking\!, both using the Min-Fanout heuristic.
Our approach is also compared to non-incremental implementations \texttt{FM} and \texttt{Simplex}.
The FMplex variants and \texttt{FM} always first employ Gaussian elimination to handle equalities.

\begin{table}[b]
    \center
    \begin{tabular}{lrrrrr|rrrrr}%
		\toprule
        {} & \multicolumn{5}{c}{SMT-LIB} & \multicolumn{5}{c}{Conjunctions}\\
		\cmidrule(lr){2-6}\cmidrule(lr){7-11}
        {} &  solved &  sat &  unsat &  TO &  MO &  solved &  sat &  unsat &  TO &  MO\\
		\midrule
        \texttt{Simplex}       &     \textbf{958} &  \textbf{527} &    \textbf{431} &      714 &      81 &    \textbf{3084} &  \textbf{777} &   \textbf{2307} &        0 &       0\\
        \texttt{FM}           &     860 &  461 &    399 &      577 &     316 &    2934 &  747 &   2187 &      107 &      43\\
        \texttt{FMplex-a-MFO}    &     814 &  432 &    382 &      840 &      99 &    2962 &  743 &   2219 &      122 &       0 \\
        \texttt{FMplex-a-MCL}    &     820 &  435 &    385 &      830 &     103 &    2965 &  742 &   2223 &      119 &       0 \\
        \texttt{FMplex-a-Rand-1} &     742 &  383 &    359 &      906 &     105 &    2806 &  668 &   2138 &      278 &       0 \\
        \texttt{FMplex-a-Rand-2} &     743 &  383 &    360 &      905 &     105 &    2823 &  671 &   2152 &      261 &       0 \\
        \texttt{FMplex-b-MFO} &     822 &  434 &    388 &      830 &     101 &    2988 &  744 &   2244 &       96 &       0 \\
        \texttt{FMplex-c-MFO} &     920 &  499 &    421 &      733 &     100 &    \textbf{3084} &  \textbf{777} &   \textbf{2307} &        0 &       0 \\
		\midrule
        \texttt{Virtual-Best} &     982 &  532 &    450 &      651 &     120 &    3084 &  777 &   2307 &        0 &       0\\
		\bottomrule
    \end{tabular}
    \caption{Number of solved instances (total and differentiated into satisfiable and unsatisfiable), timeouts (TO) and memory-outs (MO).}\label{tbl:results}
\end{table}

All solvers were tested on the SMT-LIB \cite{BarFT-SMTLIB} benchmark set for QF\_LRA containing $1753$ formulas.
As all evaluated solvers are non-incremental, we also generated conjunctions of constraints by solving each of these QF\_LRA problems using a DPLL(T) SMT solver with an \texttt{FMplex-c-MFO} theory solver backend, and extracting all conjunctions passed to it.
If the solver terminated within the time and memory limits, we sampled 10 satisfiable and 10 unsatisfiable conjunctions (or gathered all produced conjunctions if there were fewer than 10).
This amounted to 3084 (777 sat, 2307 unsat) additional benchmarks.
The experiments were conducted on identical machines with two Intel Xeon Platinum 8160 CPUs (2.1 GHz, 24 cores).
For each formula, the time and memory were limited to 10 minutes and 5 GB.

\begin{figure}[b]
    \centering
    \begin{subfigure}{0.46\textwidth}
        \includegraphics[scale=1]{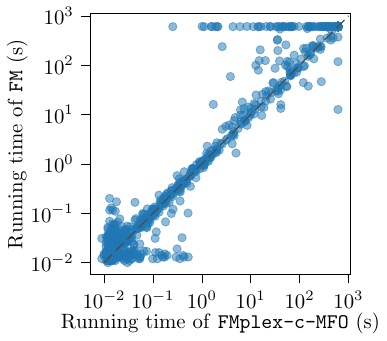}
\caption{Running times in seconds on the SMT-LIB benchmark set, \texttt{FMplex} vs \texttt{FM}.}\label{fig:runtimes-a}
    \end{subfigure}%
    \hfill
    \begin{subfigure}{0.46\textwidth}
        \includegraphics[scale=1]{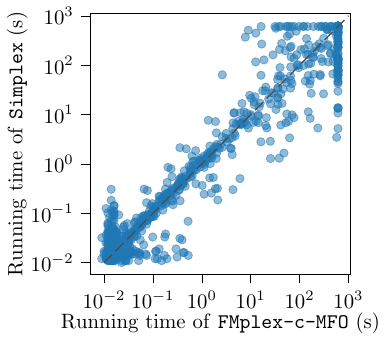}
\caption{Running times in seconds on the SMT-LIB benchmark set, \texttt{FMplex} vs \texttt{Simplex}.}\label{fig:runtimes-b}
    \end{subfigure}
    \begin{subfigure}{0.46\textwidth}
        \includegraphics[scale=1]{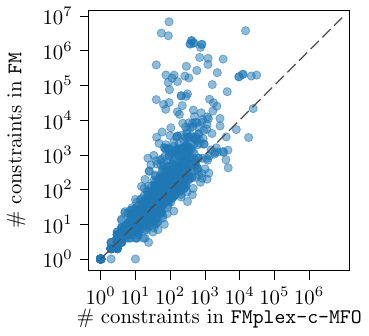}
\caption{Number of generated constraints on the conjunctive benchmark set.}\label{fig:stats-conjunctions-a}
    \end{subfigure}%
    \hfill
    \begin{subfigure}{0.46\textwidth}
        \includegraphics[scale=1]{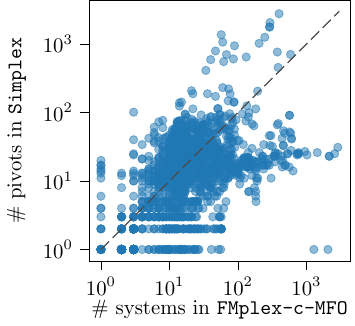}
\caption{Number of visited non-bases (intermediate systems) on the conjunctive benchmark set.}\label{fig:stats-conjunctions-b}
    \end{subfigure}
\caption{%
Scatter plots: Each dot represents a single instance.
In (A) and (B), instances at the very top or right exceeded the resource limit.
Such instances are not considered in (C) and (D).%
}\label{fig:stats-conjunctions}

\end{figure}

The results in \Cref{tbl:results} show that \texttt{Simplex} solved the most SMT-LIB instances, followed by our \texttt{FMplex-c-MFO} and then \texttt{FM}.
Interestingly, \texttt{FM} solves fewer conjunctive instances than the base version of FMplex due to higher memory consumption ($43$ memory-outs for \texttt{FM}, while the others have none).
We see that a reasonable variable heuristic makes a difference as \texttt{FMplex-a-Rand-*} perform much worse than \texttt{FMplex-a-MFO} and \texttt{FMplex-a-MCL}.
However, between the latter two, there is no significant difference.
While our first optimization used in \texttt{FMplex-b-MFO} has no big impact, the backtracking implemented in \texttt{FMplex-c-MFO} allows for solving more instances within the given resource limits.

The running times for each individual SMT-LIB instance depicted in \Cref{fig:runtimes-a} and \Cref{fig:runtimes-b} reveal that \texttt{FM} and \texttt{FMplex-c-MFO} often behave similar, but \texttt{FM} fails on a number of larger instances.
We suspect that the smaller intermediate systems of FMplex are a main factor here.
While \texttt{Simplex} is often faster than \texttt{FMplex-c-MFO} and solves 61 SMT-LIB instances not solved by \texttt{FMplex-c-MFO}, it fails to solve 23 instances on which \texttt{FMplex-c-MFO} succeeds (Of these instances, \texttt{FM} solves 3 respectively 14 instances).
Accordingly, the \texttt{Virtual-Best} of the tested solvers performs significantly better than just \texttt{Simplex}, indicating potential for a combination of \texttt{Simplex} and \texttt{FMplex-c-MFO}.

\Cref{fig:stats-conjunctions-a} compares the number of constraints generated by \texttt{FM} and \texttt{FMplex-c-MFO} on the conjunctive inputs.
Especially on larger instances, FMplex seems to be in the advantage.
Motivated by \Cref{subsec:avoiding-redundancies}, \Cref{fig:stats-conjunctions-b} compares the number of \texttt{Simplex} pivots  to the number of systems in \texttt{FMplex-c-MFO}.
We see that neither is consistently lower than the other, though \texttt{Simplex} seems to be slightly superior.
Due to the log-log scale, not shown are 1305 instances in which either measurement is 0 (920 instances for \texttt{Simplex}, 981 for \texttt{FMplex-c-MFO}).

\noindent The implementation and collected data are available at \url{https://doi.org/10.5281/zenodo.7755862}.

\section{Conclusion}
\label{sec:conclusion}

We introduced a novel method \emph{FMplex} for quantifier elimination and satisfiability checking for conjunctions of linear real arithmetic constraints.
Structural observations based on Farkas' Lemma and the Fundamental Theorem of Linear Programming allowed us to prune the elimination or the search tree.
Although the new method is rooted in the FM method, it has strong similarities with both the virtual substitution method and the simplex method.

The experimental results in the context of SMT solving show that FMplex is faster than Fourier-Motzkin and, although simplex is able to solve more instances than FMplex, there is a good amount of instances which can be solved by FMplex but cannot be solved by simplex.
While the current state of FMplex cannot compete with established solvers using the simplex method, these solvers benefit from implementations highly optimized through decades of research.

Thus, in future work we aim to combine the structural savings of FMplex with the efficient heuristic of simplex, i.e. we transfer ideas from FMplex to simplex and vice-versa.
Furthermore, we will investigate in tweaks and heuristics.
For instance, we plan to adapt the perfect elimination ordering from \cite{li2021choosing} and work on an incremental adaption for SMT solving.

In the work \cite{PromiesA24} following up our first paper \cite{Nalbach_2023} on FMplex, we show, in the context of quantifier elimination, how the resulting disjunction computed by FMplex may be reduced to a single conjunction. This process can be done with practically no overhead, solely based on the coefficients of the linear combinations (the matrix $F$ in the algorithm). 
This overcomes the main drawback in the comparison to the Fourier-Motzkin method.
Thus, FMplex may be used as replacement for Fourier-Motzkin, and is still singly exponential and relatively straight-forward to implement.
Moreover, \cite{PromiesA24} contains an experimental evaluation in the context of quantifier elimination, and there, FMplex outperforms several established tools.
We plan to also pursue this direction further.

\bibliographystyle{alphaurl}%
\bibliography{literature.bib}

\newcommand{\etalchar}[1]{$^{#1}$}
\begin{thebibliography}{JMMT20}

\bibitem[BFT16]{BarFT-SMTLIB}
Clark Barrett, Pascal Fontaine, and Cesare Tinelli.
\newblock The {S}atisfiability {M}odulo {T}heories {L}ibrary ({SMT-LIB}).
\newblock {\tt www.SMT-LIB.org}, 2016.

\bibitem[BNOT06]{barrett2006splitting}
Clark Barrett, Robert Nieuwenhuis, Albert Oliveras, and Cesare Tinelli.
\newblock Splitting on demand in {SAT} modulo theories.
\newblock In {\em International Conference on Logic for Programming Artificial Intelligence and Reasoning (LPAR'06)}, pages 512--526. Springer, 2006.
\newblock \href {https://doi.org/10.1007/11916277_35} {\path{doi:10.1007/11916277_35}}.

\bibitem[C{\'A}11]{vssmt}
Florian Corzilius and Erika {\'A}brah{\'a}m.
\newblock Virtual substitution for {SMT}-solving.
\newblock In {\em International Symposium on Fundamentals of Computation Theory (FCT'11)}, pages 360--371. Springer, 2011.
\newblock \href {https://doi.org/10.1007/978-3-642-22953-4_31} {\path{doi:10.1007/978-3-642-22953-4_31}}.

\bibitem[CKJ{\etalchar{+}}15]{corziliusSMTRAT2015}
Florian Corzilius, Gereon Kremer, Sebastian Junges, Stefan Schupp, and Erika {\'A}brah{\'a}m.
\newblock {SMT-RAT}: {A}n open source {C}++ toolbox for strategic and parallel {SMT} solving.
\newblock In {\em International Conference on Theory and Applications of Satisfiability Testing (SAT'15)}, pages 360--368. Springer, 2015.
\newblock \href {https://doi.org/10.1007/978-3-319-24318-4_26} {\path{doi:10.1007/978-3-319-24318-4_26}}.

\bibitem[Cot10]{cottonNatural2010}
Scott Cotton.
\newblock Natural domain {SMT}: A preliminary assessment.
\newblock In {\em International Conference on Formal {Modeling} and {Analysis} of {Timed} {Systems} ({FORMATS'10})}, pages 77--91. Springer, 2010.
\newblock \href {https://doi.org/10.1007/978-3-642-15297-9_8} {\path{doi:10.1007/978-3-642-15297-9_8}}.

\bibitem[Dan98]{dantzig1998linear}
George~Bernard Dantzig.
\newblock {\em Linear Programming and Extensions}, volume~48.
\newblock Princeton University Press, 1998.
\newblock \href {https://doi.org/10.1515/9781400884179} {\path{doi:10.1515/9781400884179}}.

\bibitem[DDM06]{dutertre2006integrating}
Bruno Dutertre and Leonardo De~Moura.
\newblock Integrating {S}implex with {DPLL(T)}.
\newblock {\em Computer Science Laboratory, SRI International, Tech. Rep. SRI-CSL-06-01}, 2006.

\bibitem[Far02]{farkas1902theorie}
Julius Farkas.
\newblock Theorie der einfachen {U}ngleichungen.
\newblock {\em Journal f{\"u}r die reine und angewandte Mathematik (Crelles Journal)}, 1902(124):1--27, 1902.
\newblock \href {https://doi.org/10.1515/crll.1902.124.1} {\path{doi:10.1515/crll.1902.124.1}}.

\bibitem[Fou27]{fourier1827analyse}
Jean Baptiste~Joseph Fourier.
\newblock Analyse des travaux de l’acad{\'e}mie royale des sciences pendant l’ann{\'e}e 1824.
\newblock {\em Partie math{\'e}matique}, 1827.

\bibitem[Imb90]{imbert1990redundant}
Jean-Louis Imbert.
\newblock About redundant inequalities generated by {F}ourier's algorithm.
\newblock In {\em International Conference on Artificial Intelligence: Methodology, Systems, Applications (AIMSA'90)}, pages 117--127. Elsevier, 1990.
\newblock \href {https://doi.org/10.1016/B978-0-444-88771-9.50019-2} {\path{doi:10.1016/B978-0-444-88771-9.50019-2}}.

\bibitem[Imb93]{imbert1993fourier}
Jean-Louis Imbert.
\newblock Fourier's elimination: {W}hich to choose?
\newblock In {\em International Conference on Principles and Practice of Constraint Programming (PPCP'93)}, volume~1, pages 117--129. Citeseer, 1993.

\bibitem[JMMT20]{JingComplexityFME}
Rui-Juan Jing, Marc Moreno-Maza, and Delaram Talaashrafi.
\newblock Complexity estimates for {F}ourier-{M}otzkin elimination.
\newblock In {\em International Workshop on Computer Algebra in Scientific Computing (CASC'20)}, pages 282--306. Springer, 2020.
\newblock \href {https://doi.org/10.1007/978-3-030-60026-6_16} {\path{doi:10.1007/978-3-030-60026-6_16}}.

\bibitem[KBD13]{soisimplex}
Tim King, Clark Barrett, and Bruno Dutertre.
\newblock Simplex with sum of infeasibilities for {SMT}.
\newblock In {\em International Conference on Formal Methods in Computer-Aided Design (FMCAD'13)}, pages 189--196, 2013.
\newblock \href {https://doi.org/10.1109/FMCAD.2013.6679409} {\path{doi:10.1109/FMCAD.2013.6679409}}.

\bibitem[Kha80]{KHACHIYAN198053}
Leonid~Genrikhovich Khachiyan.
\newblock Polynomial algorithms in linear programming.
\newblock {\em {USSR} Computational Mathematics and Mathematical Physics}, 20(1):53--72, 1980.
\newblock \href {https://doi.org/10.1016/0041-5553(80)90061-0} {\path{doi:10.1016/0041-5553(80)90061-0}}.

\bibitem[Kin14]{king_effective_2014}
Tim King.
\newblock {\em Effective algorithms for the satisfiability of quantifier-free formulas over linear real and integer arithmetic}.
\newblock {PhD} {Thesis}, New York University, 2014.

\bibitem[KKS14]{korovin2014towards}
Konstantin Korovin, Marek Kosta, and Thomas Sturm.
\newblock Towards conflict-driven learning for virtual substitution.
\newblock In {\em International Workshop on Computer Algebra in Scientific Computing (CASC'14)}, pages 256--270. Springer, 2014.
\newblock \href {https://doi.org/10.1007/978-3-319-10515-4_19} {\path{doi:10.1007/978-3-319-10515-4_19}}.

\bibitem[KTV09]{korovinConflict2009}
Konstantin Korovin, Nestan Tsiskaridze, and Andrei Voronkov.
\newblock Conflict resolution.
\newblock In {\em International Conference on Principles and {Practice} of {Constraint} {Programming} (CP'09)}, pages 509--523. Springer, 2009.
\newblock \href {https://doi.org/10.1007/978-3-642-04244-7_41} {\path{doi:10.1007/978-3-642-04244-7_41}}.

\bibitem[KV11]{korovinSolving2011}
Konstantin Korovin and Andrei Voronkov.
\newblock Solving systems of linear inequalities by bound propagation.
\newblock In {\em International Conference on Automated {Deduction} ({CADE}'23)}, pages 369--383. Springer, 2011.
\newblock \href {https://doi.org/10.1007/978-3-642-22438-6_28} {\path{doi:10.1007/978-3-642-22438-6_28}}.

\bibitem[Lem54]{lemke1954dual}
Carlton~E. Lemke.
\newblock The dual method of solving the linear programming problem.
\newblock {\em Naval Research Logistics Quarterly}, 1(1):36--47, 1954.
\newblock \href {https://doi.org/10.1002/nav.3800010107} {\path{doi:10.1002/nav.3800010107}}.

\bibitem[LW93]{loosApplyingLinearQuantifier1993}
R{\"u}diger Loos and Volker Weispfenning.
\newblock Applying linear quantifier elimination.
\newblock {\em The Computer Journal}, 36(5):450--462, 1993.
\newblock \href {https://doi.org/10.1093/comjnl/36.5.450} {\path{doi:10.1093/comjnl/36.5.450}}.

\bibitem[LXZZ21]{li2021choosing}
Haokun Li, Bican Xia, Huiying Zhang, and Tao Zheng.
\newblock Choosing the variable ordering for cylindrical algebraic decomposition via exploiting chordal structure.
\newblock In {\em International Symposium on Symbolic and Algebraic Computation (ISSAC'21)}, pages 281--288. ACM, 2021.
\newblock \href {https://doi.org/10.1145/3452143.3465520} {\path{doi:10.1145/3452143.3465520}}.

\bibitem[LY{\etalchar{+}}84]{luenberger1984linear}
David~G Luenberger, Yinyu Ye, et~al.
\newblock {\em Linear and Nonlinear Programming}, volume 2nd edition.
\newblock Springer, 1984.
\newblock \href {https://doi.org/10.1007/978-3-030-85450-8} {\path{doi:10.1007/978-3-030-85450-8}}.

\bibitem[MKS09]{mcmillanGeneralizing2009}
Kenneth~L. McMillan, Andreas Kuehlmann, and Mooly Sagiv.
\newblock Generalizing {DPLL} to richer logics.
\newblock In {\em International Conference on Computer {Aided} {Verification} {(CAV'09)}}, pages 462--476. Springer, 2009.
\newblock \href {https://doi.org/10.1007/978-3-642-02658-4_35} {\path{doi:10.1007/978-3-642-02658-4_35}}.

\bibitem[Mot36]{motzkin1936beitrage}
Theodore~Samuel Motzkin.
\newblock {\em Beitr{\"a}ge zur {T}heorie der linearen {U}ngleichungen}.
\newblock Azriel, 1936.

\bibitem[N{\'A}K21]{nalbach2021extending}
Jasper Nalbach, Erika {\'A}brah{\'a}m, and Gereon Kremer.
\newblock Extending the fundamental theorem of linear programming for strict inequalities.
\newblock In {\em International Symposium on Symbolic and Algebraic Computation (ISSAC'21)}, pages 313--320. ACM, 2021.
\newblock \href {https://doi.org/10.1145/3452143.3465538} {\path{doi:10.1145/3452143.3465538}}.

\bibitem[Nip08]{nipkowLinearQuantifierElimination2008a}
Tobias Nipkow.
\newblock Linear quantifier elimination.
\newblock In {\em Internation Joint Conference on Automated Reasoning (IJCAR'08)}, pages 18--33. Springer, 2008.
\newblock \href {https://doi.org/10.1007/978-3-540-71070-7_3} {\path{doi:10.1007/978-3-540-71070-7_3}}.

\bibitem[NP{\'{A}}K23]{Nalbach_2023}
Jasper Nalbach, Valentin Promies, Erika {\'{A}}brah{\'{a}}m, and Paul Kobialka.
\newblock {FMplex}: A novel method for solving linear real arithmetic problems.
\newblock In {\em International Symposium on Games, Automata, Logics, and Formal Verification (GandALF'23)}, volume 390 of {\em EPTCS}, pages 16--32, 2023.
\newblock \href {https://doi.org/10.4204/EPTCS.390.2} {\path{doi:10.4204/EPTCS.390.2}}.

\bibitem[P{\'{A}}24]{PromiesA24}
Valentin Promies and Erika {\'{A}}brah{\'{a}}m.
\newblock A divide-and-conquer approach to variable elimination in linear real arithmetic.
\newblock In Andr{\'{e}} Platzer, Kristin~Yvonne Rozier, Matteo Pradella, and Matteo Rossi, editors, {\em International Symposium on Formal Methods ({FM}'24)}, volume 14933 of {\em LNCS}, pages 131--148. Springer, 2024.
\newblock \href {https://doi.org/10.1007/978-3-031-71162-6_7} {\path{doi:10.1007/978-3-031-71162-6_7}}.

\bibitem[Wei97]{weispfenning1997quantifier}
Volker Weispfenning.
\newblock Quantifier elimination for real algebra — {T}he quadratic case and beyond.
\newblock {\em Applicable Algebra in Engineering, Communication and Computing}, 8(2):85--101, 1997.
\newblock \href {https://doi.org/10.1007/s002000050055} {\path{doi:10.1007/s002000050055}}.

\end{thebibliography}

\end{document}